\documentclass[a4paper,11pt]{autart}
\usepackage[T1]{fontenc}
\usepackage[utf8]{inputenc}
\usepackage[english]{babel}

\usepackage{setspace}
\usepackage{placeins}

\usepackage{psfrag}
\usepackage{bm}
\usepackage{float}
\usepackage{subfig}

\usepackage{graphicx}
\usepackage{amsmath,amssymb}
\usepackage{mathrsfs}
\usepackage{picins}
\usepackage{multirow}

\usepackage{psfrag}
\usepackage{bm}
\usepackage{algorithmic,algorithm}
\usepackage{xspace}

\usepackage{wasysym}

\newcommand{\rf}[1]{(\ref{#1})}

\DeclareMathOperator*{\argmin}{arg\,min}
\DeclareMathOperator*{\argmax}{arg\,max}
\DeclareMathOperator*{\rank}{rank}
\DeclareMathOperator*{\minimize}{minimize}
\DeclareMathOperator*{\sign}{sign}

\DeclareMathOperator*{\im}{im}

\DeclareMathOperator*{\conv}{conv}
\DeclareMathOperator*{\Prob}{Prob}
\DeclareMathOperator*{\supp}{supp}
\DeclareMathOperator*{\suchthat}{ s.t. }

\DeclareMathOperator*{\tr}{tr}
\DeclareMathOperator*{\col}{col}

\DeclareRobustCommand*\eg{\emph{e.g.,}\@\xspace}
\DeclareRobustCommand*\etc{\emph{etc.}\@\xspace}
\DeclareRobustCommand*\ie{\emph{i.e.,}\@\xspace}

\renewcommand{\Re}{\mathbb{R}}

\usepackage{bbm}
\renewcommand{\paragraph}[1]{\smallskip\noindent\textbf{#1.} }

\renewcommand\footnotemark{}

\begin{document} 

\setstretch{1}
\parskip 1pt  

\begin{frontmatter}

\title{Analysis of A  Nonsmooth Optimization Approach to Robust Estimation\thanksref{footnoteinfo}}   

\thanks[footnoteinfo]{This paper was not presented at any IFAC 
meeting. Corresponding author L. Bako. Tel.: +33 472 186 452.}

\author[AMP-ECL-UnivLyon]{Laurent Bako}
\ead{laurent.bako@ec-lyon.fr}
\address[AMP-ECL-UnivLyon]{Laboratoire Amp\`{e}re -- Ecole Centrale de Lyon -- Universit\'{e} de Lyon, France}

\author[Linkoping,UCBerkeley]{Henrik Ohlsson}
\address[Linkoping]{Department of Electrical Engineering, Link\"{o}ping University, SE-581 83 Link\"{o}ping, Sweden}
\address[UCBerkeley]{Department of Electrical Engineering and Computer Sciences, University of California at Berkeley, CA, USA}
\ead{ohlsson@berkeley.edu}

\maketitle
\begin{abstract} 
In this paper,  we consider the problem of identifying a linear map from measurements which are subject to intermittent and arbitarily large   errors. This is a fundamental problem  in many estimation-related  applications  such as  fault detection, state estimation in lossy networks, hybrid system identification, robust estimation, etc. The problem is hard because it exhibits some intrinsic combinatorial features. Therefore, obtaining an effective solution necessitates relaxations that are both solvable at a reasonable cost and effective in the sense that they can return the true parameter vector. The current paper  discusses a nonsmooth convex optimization  approach and provides a new  analysis of its behavior. In particular, it is shown that under appropriate conditions on the data, an exact estimate can be recovered from data corrupted by a large (even infinite) number of gross errors. 
\end{abstract}
\end{frontmatter}
\section{Introduction}\label{sec:Intro}

\subsection{Problem and motivations}
We consider a linear measurement model of the form 
\begin{equation}\label{eq:model}
	y_t=x_t^\top\theta^o+f_t+e_t
\end{equation}
where $y_t\in \Re$ is the measured signal, $x_t\in \Re^n$ the regression vector,  
$\left\{e_t\right\}$  a sequence of \textit{zero-mean and bounded} errors (\eg measurement noise, model mismatch, uncertainties, \etc) and 
 $\left\{f_t\right\}$  a sequence of \textit{intermittent and arbitrarily large} errors.
Assume that we observe the sequences $\left\{x_t\right\}_{t=1}^N$ and  $\left\{y_t\right\}_{t=1}^N$ and would like to compute the parameter vector $\theta^o$ from these observations. We are interested in doing so without knowing any of the sequences   $\left\{f_t\right\}$ and $\left\{e_t\right\}$. We do however  make the following assumptions:
\begin{itemize}
	\item $\left\{e_t\right\}$ is a bounded sequence. 
	\item $\left\{f_t\right\}$ is a sequence containing zeros and intermittent gross errors with (possibly) arbitrarily large magnitudes.
\end{itemize}
This is an important estimation problem arising in many situations such as fault detection \cite{Ozay11-CDC-ECC,Chen11-IFAC}, hybrid system identification \cite{Garulli12-SYSID}, subspace clustering \cite{Vidal10-SPM,Bako14-SPL}, error correction in communication networks \cite{Candes06-IT}. The case when $\left\{f_t\right\}$ is zero  and $\left\{e_t\right\}$ is a Gaussian process has been well-studied in linear system identification theory (see, \eg the text books \cite{LjungBook,Soderstrom-Book89}). 
A less studied, but very relevant scenario in the system identification community, is when the additional perturbation $\{f_t\}$ in \eqref{eq:model} is nonzero and contains intermittent and arbitrarily large errors  \cite{Candes06-IT,Sharon09-ACC,Mitra13-SP,Xu14-Automatica}.   
It is worth noticing the difference with the problem studied in the field of
compressive sensing  \cite{Candes06-IT,Donoho:06,Candes08-SPM}. In compressive sensing, the sought parameter
vector is assumed sparse and the measurement noise $\left\{e_t\right\}$, often Gaussian or
bounded. Here, no assumptions are made concerning sparsity of $\theta^o$. We will, in this contribution, study essentially  the case when  the data is noise-free (\ie $e_t=0$ for all $t$) and $\{f_t\}$ is a  sequence with intermittent gross errors. We will derive conditions for perfect recovery and point to effective algorithms for computing $\theta^o$. 
In the second part of the paper, the model assumption is relaxed to allow both $e_t$ and $f_t$ to be  simultaneously nonzero. Note that this might be a more realistic scenario since most applications have measurement noise. \\
For illustrative purposes, let us discuss briefly  some applications where a model of the form \eqref{eq:model} is of interest.

\paragraph{Switched linear system identification}
A discrete-time Multi-Input Single-Output (MISO) Switched Linear System (SLS) can be written in
the form
\begin{equation}\label{eq:SARX}
	   y_t=x_t^\top{\theta_{\sigma_t}^o} +e_t,
\end{equation}
where 
 $x_t\in \Re^n$ is the regressor at time $t\in \mathbb{Z}_+$ defined by
\begin{equation}\label{eq:Regressor}
	 x_t=\big[\begin{matrix}y_{t-1}&\cdots & y_{t-n_a} & u_t^\top & u_{t-1}^\top & \cdots & u_{t-n_b}^\top  \end{matrix}\big]^\top,
\end{equation}
where $u_t\in \Re^{n_u}$ and  $y_t\in \Re$ denote respectively the input and the
output of the system. The integers $n_a$ and $n_b$ in
\eqref{eq:Regressor} are  the maximum  output and input lags (also called the orders of the system).   
 $\sigma_t\in \left\{1,\ldots,s\right\}$ is the
 discrete mode (or discrete state) indexing the active
 subsystem at time $t$; it  is in general assumed \textit{unobserved}. $\theta_{\sigma_t}^o\in  \Re^{n}$, $n=n_a+n_bn_u$, is the parameter vector (PV) associated with the mode $\sigma_t$. For  $\theta^o\in \left\{\theta_1^o,\ldots,\theta_s^o\right\}$, the  Switched Auto-Regressive eXogenous (SARX) model \eqref{eq:SARX} can be written in the form  \eqref{eq:model}, with unknown $f_t$ of the following structure $f_t=x_t^\top(\theta_{\sigma_t}^o-\theta_{}^o )$. For a background on hybrid system identification, we refer to the references 
 \cite{Paoletti07,Garulli12-SYSID,Vidal03,Bako11-Automatica,Maruta11-CDC,Ozay12-TAC,Ohlsson13-Automatica}. 

\paragraph{Identification from faulty data}
A model of the form \eqref{eq:model} also arises when one has to identify a linear dynamic system which is subject to intermittent sensor faults. This is  the case in general  when the data are transmitted over a communication network \cite{Candes06-IT,Ozay11-CDC-ECC}. Model \eqref{eq:model} is suitable for such situations and the sequence $\left\{f_t\right\}$ then models occasional data packets losses or potential outliers. 
More precisely,  a dynamic MISO system with process faults can be directly written in the form \eqref{eq:model}. In the case of sensor faults,  the faulty model might be defined by 
$$ 
\left\{\begin{aligned}
	& \bar{y}_t=\bar{x}_t^\top{\theta^o} +e_t \\
	& y_t=\bar{y}_t+w_t
\end{aligned} \right.
$$
where $y_t\in \Re$ is the \textit{observed output} which is affected by the fault $w_t$ (assumed to be nonzero only occasionally) ; $\bar{x}_t$ is defined as in \eqref{eq:Regressor} from the \textit{known input} $u_t$ and  the \textit{unobserved  output} $\bar{y}_t$.  We can rewrite the faulty model exactly in the form \eqref{eq:model} with 
$f_t=w_t-\big[\begin{matrix}w_{t-1} & \cdots & w_{t-n_a} \end{matrix}\big]\theta^o.$ Sparsity of $\left\{w_t\right\}$ induces sparsity of $\left\{f_t\right\}$ but in a lesser extent. 

\paragraph{State estimation in the presence of intermittent errors}
Considering a MISO dynamic system with state dynamics described by
$z_{t+1}=Az_t+Bu_t$ and observation equation $\tilde{y}_t=C^\top
z_t+f_t$,  $(A,B,C)$ being known matrices of appropriate dimensions, and $\left\{f_t\right\}$  a sparse sequence of possibly very large  errors,  the finite horizon state estimation problem reduces to the estimation of the initial state $z_0=\theta$. We get a model of the form  \eqref{eq:model} by setting 
$y_t=\tilde{y}_t-C^\top \Delta_t \bar{u}_t$ and  $x_t=(A^t)^\top C$,
with $\Delta_t=\big[\begin{matrix}A^{t-1}B & \cdots & AB &
  B\end{matrix}\big]$, $\bar{u}_t =\big[\begin{matrix}u_0^\top &
  \cdots & u_{t-1}^\top\end{matrix}\big]^\top$. 
Examples of relevant works are those reported in \cite{Bako13-SLC,Fawzi14-TAC}.		
In this latter application, it can however be noted that the dataset $\left\{x_t\right\}$ may not be
generic enough.  \footnote{In this paper, the term genericity for a dataset characterizes a notion of linear independence.  For example, a set of $N>n$  data points in general linear position in $\Re^n$ is more generic than a set of data points  contained in one subspace. We will introduce different quantitative measures of data genericity in the sequel (see Definition \ref{def:genericity} and Theorem \ref{prop:sufficient-condition}).} 

\paragraph{Connection to robust statistics}
Indeed, the problem of identifying the parameters from model \eqref{eq:model} under the announced assumptions can be viewed as a robust regression problem where the  nonzero elements in the sequence  $\left\{f_t\right\}$ are termed outliers. As such, it has received a lot of attention in the robust statistics literature (see, e.g., \cite{Huber-Book-09,Rousseeuw05-Book,Maronna06-Book} for an overview). Examples of methods to tackle the robust estimation problem include the least absolute deviation \cite{Huber87-L1}, the least median of squares \cite{Rousseeuw84}, the least trimmed squares \cite{Rousseeuw05-Book}, the M-estimator \cite{Huber-Book-09}, etc. Most of these estimators come with an analysis in terms of the  breakdown point  \cite{Hampel71,SeberBook03}, a measure of the (asymptotic) minimum proportion of  points  which cause an estimator to be unbounded if they were to be arbitrarily corrupted by gross errors.  The current report focuses on the analysis of a nonsmooth convex optimization approach which includes the least absolute deviation method as a particular case corresponding to the situation when the output in \eqref{eq:model} is a scalar.  The analysis approach taken in the current paper is different in the following sense. 
\begin{itemize}
	\item In robust statistics the quality of an estimator is measured by its breakdown point. The higher the breakdown point, the better. The available analysis is therefore directed to  determining a sort of absolute robustness: how many outliers (expressed in proportion of the total number of samples) cause the estimator to become unbounded.  
	\item Here, the question of robust performance of the estimator is posed differently. We are interested in estimating the maximum number of outliers that a nonsmooth-optimization-based estimator can accommodate while still returning the exact value one would obtain in the absence of any outlier. This is more related to the traditional view developed in compressive sensing. 
\end{itemize}

\paragraph{Contributions of this paper}
 One promising method for estimating model \eqref{eq:model} is by
 nonsmooth convex optimization as suggested in
 \cite{Candes06-IT,Sharon09-ACC,Bako11-Automatica,Mitra13-SP,Xu14-Automatica}. More precisely,
 inspired by the recent theory of compressed sensing
 \cite{Candes06-IT,Donoho:06,Candes08-SPM}, the idea is to minimize a
 nonsmooth (and non differentiable) sum-of-norms objective function
 involving the fitting errors. Under noise-free assumption, such a cost function has the nice
 property that it is able to provide the true parameter vector in the
 presence of arbitrarily large errors $\left\{f_t\right\}$ provided
 that the number of nonzero errors is small in some sense. Of course, when the data are corrupted simultaneously by the noise $\left\{e_t\right\}$ and the gross errors $\left\{f_t\right\}$, the recovery cannot be exact any more. It is however expected (as Proposition \ref{prop:error-bound} and simulations tend to suggest) that  the estimate will still be close to the true one.\\ 
The current paper intends to present a new analysis of the nonsmooth optimization approach and provide some elements for further  understanding its behavior. The line of analysis goes from a full characterization of the nonsmooth optimization based  estimator (both for SISO and MIMO systems) to the study of robustness to outliers including in the presence of dense noise. With respect to relevant works  \cite{Candes06-IT,Sharon09-ACC,Bako11-Automatica,Mitra13-SP,Xu14-Automatica}, we derive new bounds on the number of outliers (in the least favorable situations) that the estimator is capable to accommodate. It is emphasized that a quite broad spectrum of such bounds can be derived based on the basic characterization of the nonsmooth identifier. Note however that evaluating numerically the tightest of these bounds is a high computational process while less tight bounds have a more affordable complexity. Some of the bounds developed in this contribution meet both relative tightness requirement and computability in polynomial time (see the bound based on $\xi(X)$ in Theorem \ref{prop:sufficient-condition}). Finally, the paper show how the results derived in the first part for $\ell_1$-norm estimator when applied to the estimation of SISO systems are generalizable   to multivariable systems. 

\paragraph{Outline of this paper} The outline of the paper is as follows. We start in Section \ref{subsec:Sparse-Opt} by viewing  the nonsmooth optimization  as the convex relaxation of a (ideal) combinatorial  $\ell_0$-"norm" formulation.  We then derive in Section \ref{subsec:L1-characterization} necessary and sufficient conditions for optimality. Based on those conditions we establish in Section \ref{subsec:Sufficient-Conditions}  new sufficient conditions for exact recovery of the true parameter vector in \eqref{eq:model}. The noisy case is treated in Section  \ref{subsec:Noise}. Section \ref{sec:multivariable} presents a generalization of the earlier discussions to multi-output systems. Finally, numerical experiments are described in Section \ref{sec:simulation} and concluding remarks are given in Section \ref{sec:conclusion}.

\subsection{Notations}
Let $\mathbb{I}=\left\{1,\ldots,N\right\}$ be the index set of the measurements. 
For any $\theta\in \Re^{n}$, define a partition of the set of indices $\mathbb{I}$ by 
$\mathbb{I}^-(\theta)=\left\{t\in \mathbb{I}:\theta^\top x_t-y_t<0\right\}$,  
$\mathbb{I}^+(\theta)=\left\{t \in \mathbb{I}:\theta^\top x_t-y_t>0\right\}$,  
$\mathbb{I}^0(\theta)=\left\{t\in \mathbb{I}:\theta^\top x_t-y_t=0\right\}$. 
\textit{Cardinality of a finite set.} Throughout the paper, whenever $\mathcal{S}$ is a finite set, the notation $\left|\mathcal{S}\right|$ will refer to the cardinality of $\mathcal{S}$. However, for a real number $x$, $\left|x\right|$ will denote the absolute value of $x$.\\
\textit{Submatrices and subvectors.}
Let $X = \big[\begin{matrix}x_1 & x_2 & \cdots & x_N\end{matrix}\big]\in \Re^{n\times N}$ be the matrix formed with the available regressors $\left\{x_t\right\}_{t=1}^N$.  If $I\subset \mathbb{I}$, the notation $X_I$ denotes a matrix in $\Re^{n\times \left|I\right|}$ formed with the columns of $X$ indexed by $I$.
Likewise, with $\bm{y}=\big[\begin{matrix}y_1 & y_2 & \cdots & y_N\end{matrix}\big]^\top\in \Re^N$, $\bm{y}_I$ is the vector in $\Re^{\left|I\right|}$ formed with the entries of $\bm{y}$ indexed by $I$. 
We will use the convention that $X_I=0\in \Re^n$ (resp. $\bm{y}_I=0\in \Re$) when the  index set $I$ is empty. \\
\textit{Vector norms.}
 $\left\|\cdot\right\|_p$, $p=1,2, \ldots, \infty$,  denote the usual $p$-norms for vectors defined for any vector  $z=\big[\begin{matrix}z_1 & \cdots & z_N\end{matrix}\big]^\top\in \Re^N$, by  $\left\|z\right\|_p=\left(\left|z_1\right|^p+\cdots+\left|z_N\right|^p\right)^{1/p}$. Note that $\left\|z\right\|_\infty=\max_{i=1,\ldots,N}\left|z_i\right|$. 
The $\ell_0$ "norm" of $z$ is defined to be the number of nonzero entries in $z$, i.e., $\left\|z\right\|_0=\left|\left\{i: z_i\neq 0\right\}\right|$. \\
\textit{Matrix norms.} The following matrix norms will be used: $\left\|\cdot\right\|_{p}$, $p=1,2, \ldots, \infty$, $\left\|\cdot\right\|_{2,\col}$, $\left\|\cdot\right\|_{2,\infty}$. They are defined as follows:  for a matrix $A=\big[\begin{matrix}a_1 & \cdots & a_N\end{matrix}\big]\in \Re^{n\times N}$ with $a_i\in \Re^n$, 
{\setstretch{.8}
$$\begin{aligned}
	&\left\|A\right\|_p=\sup_{x\in \Re^N,\left\|x\right\|_p=1}\left\|Ax\right\|_p,  
	\quad \left\|A\right\|_{2,\col}=\sum_{i=1}^N \left\|a_i\right\|_2,  \\
	&  \left\|A\right\|_{2,\infty}=\max_{i=1,\ldots,N}\left\|a_i\right\|_2. 
	\end{aligned}
	$$
	}
\vspace*{-\baselineskip}
\section{Nonsmooth optimization for the estimation problem}
\subsection{Sparse optimization}\label{subsec:Sparse-Opt}
The main idea for identifying the parameter vector $\theta^o$ from \eqref{eq:model} is by solving a sparse optimization problem, that is, a problem which involves the minimization of the number of nonzeros entries  in the error vector.  To be more specific, assume for the time being that the error sequence $\left\{e_t\right\}$ is identically equal to zero. Consider a candidate parameter vector $\theta\in \Re^n$ and  let 
$$\phi(\theta) =\bm{y}-X^\top \theta, $$
where $\bm{y}=\big[\begin{matrix}y_1  & \cdots & y_N\end{matrix}\big]^\top$, $X = \big[\begin{matrix}x_1  & \cdots & x_N\end{matrix}\big]$,
be the fitting error vector induced by $\theta$ on the experimental data. 
Then the vector  $\theta^o$ can naturally be searched for by minimizing an $\ell_0$ objective function,  
\begin{equation}\label{eq:L0-Pbm}
	\minimize_{\theta\in \Re^n}\left\|\phi(\theta)\right\|_0
\end{equation}
where $\left\|\cdot\right\|_0$ denotes the $\ell_0$ pseudo-norm which counts the number of nonzero entries. Because problem \eqref{eq:L0-Pbm} aims at making the error $\phi(\theta)$ sparse by minimizing the number of nonzero elements (or maximizing the number of zeros), it is sometimes called a sparse optimization problem \cite{Bako11-Automatica}. 
As can be intuitively guessed, the recoverability of the true
parameter vector $\theta^o$ from \eqref{eq:L0-Pbm} depends naturally
on some properties  of the available data. This is outlined by the following lemma. 
\begin{lem}[A  sufficient condition for $\ell_0$ recovery]
\label{lem:L0-solution-set}
Assume that $\left\{e_t\right\}$ is equal to zero and 
let $\bm{f}=\big[\begin{matrix}f_1&\cdots &
  f_N\end{matrix}\big]^\top.$ 
Assume that  for any $I\subset \mathbb{I}$  with $\left|I\right|>n$, $\bm{f}_I\notin \im(X_I^\top)$ whenever $\bm{f}_I\neq 0$, with $\im(\cdot)$ referring here to  range space. 
Then provided   $\left|\mathbb{I}^0(\theta^o)\right|>n$, it holds that 
\begin{equation}
\label{eq:L0minSet}
	\theta^o\in \argmin_{\theta} \left\|\phi(\theta)\right\|_0. 
\end{equation}
\end{lem}
\begin{pf}
We proceed by contradiction. Assume that \eqref{eq:L0minSet} does not hold, \ie $\min_{\theta}\left\|\phi(\theta)\right\|_0< \left\|\phi(\theta^o)\right\|_0$. Then, by letting $\theta^m$ be any vector in $\argmin_{\theta} \left\|\phi(\theta)\right\|_0$, the above inequality translates into  $\left|\mathbb{I}^0(\theta^m)\right|>\left|\mathbb{I}^0(\theta^o)\right|>n$. It follows that $\bm{f}_{\mathbb{I}^0(\theta^m)}\neq 0$ because $\left|\mathbb{I}^0(\theta^o)\right|=\left|\left\{t\in \mathbb{I}:f_t= 0\right\}\right|$ is the exact (largest) number of zero elements in the sequence $\left\{f_t\right\}_{t=1}^N$. Note also that we have necessarily $\left|\mathbb{I}^0(\theta^m)\right|>n$. On the other hand, with $\bm{y}_{\mathbb{I}^0(\theta^m)}=X_{\mathbb{I}^0(\theta^m)}^\top \theta^m=X_{\mathbb{I}^0(\theta^m)}^\top \theta^o+\bm{f}_{\mathbb{I}^0(\theta^m)}$, it can be seen that $\bm{f}_{\mathbb{I}^0(\theta^m)}=X_{\mathbb{I}^0(\theta^m)}^\top \left(\theta^m-\theta^o\right)\in \im(X_{\mathbb{I}^0(\theta^m)}^\top)$. This, together with $\bm{f}_{\mathbb{I}^0(\theta^m)}\neq 0$, constitutes a contradiction to the assumption of the Lemma. Hence, \eqref{eq:L0minSet} holds as claimed.  \qed 
\end{pf}%
\vspace*{-\baselineskip}
Lemma \ref{lem:L0-solution-set} specifies a condition involving both $X$ and $\bm{f}$ and under which $\theta^o$ lies in the solution set but it does not ensure that $\theta^o$ will be recovered uniquely from data. 
Before proceeding further, we  recall from \cite{Bako11-Automatica} a sufficient condition under which $\theta^o$ is the unique solution to \eqref{eq:L0-Pbm}. 
\begin{defn}[\cite{Bako11-Automatica} An integer measure of genericity]
\label{def:genericity}
Let $X\in \Re^{n\times N}$ be a data matrix satisfying $\rank(X)=n$.
The $n$-genericity index of $X$ denoted   $\nu_n(X)$,   is defined as 
the minimum integer $m$ such that any $n\times m$ submatrix of $X$  has rank $n$,  
{
\begin{equation}\label{eq:Nu_n_X}
\nu_n(X)=\min\Big\{m:\forall \: \mathcal{S}\subset I \mbox{ with }\left|\mathcal{S}\right|=m, \:  \rank(X_\mathcal{S})=n\Big\}.
\end{equation}
}
\end{defn}
\begin{thm}[\cite{Bako11-Automatica} Sufficient condition for $\ell_0$ recovery]\label{thm:Uniqueness_L0}
Assume that the sequence $\left\{e_t\right\}$ in \eqref{eq:model} is identically equal to zero. If the sequence $\left\{f_t\right\}$ in \eqref{eq:model} contains enough zero values in the sense that    
\begin{equation}\label{eq:uniqueness_L0}
\left|\mathbb{I}^0(\theta^o)\right| = \left|\big\{t\in \mathbb{I}: f_t=0\big\} \right|\geq \dfrac{N+\nu_n(X)}{2}, 
\end{equation}
then $\theta^o$  is the unique solution to the $\ell_0$-norm minimization problem \eqref{eq:L0-Pbm}. 
\end{thm}
In other words, if the number of nonzero gross errors $\left\{f_t\right\}$ affecting the data generated by \eqref{eq:model} does not exceed the threshold $(N-\nu_n(X))/2$, then $\theta^o$ can be exactly recovered by solving \eqref{eq:L0-Pbm}. Unfortunately, this problem is a hard combinatorial optimization problem. A  tractable solution can be obtained by relaxing the $\ell_0$-norm into its best convex approximant, the $\ell_1$-norm.  
Doing this substitution  in \eqref{eq:L0-Pbm} gives 
\begin{equation}\label{eq:L1-Cost}
	\minimize_{\theta\in \Re^n}\left\|\phi(\theta)\right\|_1
\end{equation}
with $\left\|\phi(\theta)\right\|_1=\sum_{t=1}^N\left|y_t-\theta^\top x_t\right|$. 
The latter problem is termed a nonsmooth convex optimization problem \cite[Chap. 3]{Nesterov04-Book} because the objective function is convex but non-differentiable.  Compared to \eqref{eq:L0-Pbm}, problem \eqref{eq:L1-Cost} has the advantage of being convex and can hence be efficiently solved by many existing numerical solvers, \eg \cite{Grant-Boyd_CVX}. 
Note further that it can  be written as a linear programming problem. 
The $\ell_1$ relaxation process has been intensively used in the compressed sensing literature \cite{Donoho03-NAS} for approaching the sparsest solution of an underdetermined set of linear equations. In the robust statistics literature as surveyed above, \eqref{eq:L1-Cost} corresponds to  a well-known estimator referred to as the least absolute deviation estimator \cite{Rousseeuw84}. 
As will be shown next, the underlying reason why problem \eqref{eq:L1-Cost} can obtain the true parameter vector despite the presence of gross perturbations $\left\{f_t\right\}$ is related to its nonsmoothness.

\subsection{Solution to the $\ell_1$ problem}
\label{subsec:L1-characterization}

There is a wealth of analysis in the literature of compressed sensing investigating under which conditions some problems\footnote{Those problems look for the sparsest solution to an underdetermined set of linear equations. As such they are similar but different to the problem studied in the current paper. Note that the process of converting problems \eqref{eq:L0-Pbm} and \eqref{eq:L1-Cost} into the format treated in compressed sensing yields a system of linear equations which is much less underdetermined.} 
of similar structure as  \eqref{eq:L0-Pbm} and \eqref{eq:L1-Cost} can yield the same solution. This analysis is mainly based on the concepts of mutual coherence \cite{Donoho03-NAS} and the Restricted Isometry Property \cite{Candes05-IT}. 
Here, we shall propose a parallel but different analysis for the robust estimation problem. We start by characterizing the solution to the $\ell_1$-norm problem \eqref{eq:L1-Cost}. 
\begin{thm}[Solution to the $\ell_1$ problem]\label{thm:equivalence}
A vector $\theta^\star\in \Re^n$ solves the $\ell_1$-norm problem \eqref{eq:L1-Cost} if and only if any of the following equivalent statements hold: 
\begin{enumerate}
\item[S1.] There exist some  numbers $\lambda_t\in \left[-1,\; 1\right]$, $t\in \mathbb{I}^0(\theta^\star)$, such that\footnote{Eq. \eqref{eq:iff-Cond} should be understood here with the implicit convention  that any of the three terms is equal to zero whenever the corresponding index set is empty.  } 
\begin{equation}\label{eq:iff-Cond}
\sum_{t\in \mathbb{I}^+(\theta^\star)} x_t-\sum_{t\in \mathbb{I}^-(\theta^\star)}x_t= \sum_{t\in \mathbb{I}^0(\theta^\star)}\lambda_t x_t.
\end{equation}
\item[S2.] For any $\eta \in \Re^n$, 
\begin{equation}\label{eq:iff-Cond2}
	\Big|\sum_{t\in \mathbb{I}^+(\theta^\star)}\eta^\top x_t-\sum_{t\in \mathbb{I}^-(\theta^\star)}\eta^\top x_t\Big|\leq \sum_{t\in \mathbb{I}^0(\theta^\star)}\big|\eta^\top x_t\big|.
\end{equation}
 \item[S3.] The optimal value of the optimization problem 
\begin{equation}\label{eq:min_alpha<1}
	\min_{\alpha}\left\|\alpha\right\|_\infty \: \mbox{ subject to } z=X_{\mathbb{I}^0(\theta^\star)}\alpha,
\end{equation}
where $z = \sum_{t\in \mathbb{I}^+(\theta^\star)}x_t-\sum_{t\in \mathbb{I}^-(\theta^\star)} x_t$, $\alpha\in \Re^{\left|\mathbb{I}^0(\theta^\star)\right|}$, 
is less than or equal to $1$.
\end{enumerate} 
\noindent Moreover,  the solution $\theta^\star$ is unique if and only if any of the following statements is true:
\begin{enumerate}
	\item[S1'.] \eqref{eq:iff-Cond} holds and $\rank(X_S)=n$ where  
	$$S=\left\{t\in \mathbb{I}^0(\theta^\star): \left|\lambda_t\right|<1\right\}.$$ 
	\item[S2'.] \eqref{eq:iff-Cond2} holds with strict inequality symbol for all $\eta\in \Re^n$, $\eta\neq 0$.
\end{enumerate}
\end{thm}
\vspace*{-\baselineskip}
\begin{pf}
\paragraph{Proof of S1}
Since $\left\|\phi(\theta)\right\|_1$ is a proper convex function, it has a non empty subdifferential \cite{Rockafellar97}. 
The necessary and sufficient condition for $\theta^\star$ to be a solution of \eqref{eq:L1-Cost} is then 
\begin{equation}
	0\in \partial\left\|\phi(\theta^\star)\right\|_1,
\end{equation}
where the notation $\partial$ refers to subdifferential with respect to $\theta$. Indeed, using additivity of subdifferentials, it is straightforward to write 
\begin{equation}\label{eq:Subdifferential}
	\partial \left\|\phi(\theta^\star)\right\|_1 = \!\!\!\!\sum_{t\in \mathbb{I}^+(\theta^\star)}x_t-\!\!\!\!\sum_{t\in \mathbb{I}^-(\theta^\star)}x_t+\!\!\!\!\sum_{t\in \mathbb{I}^0(\theta^\star)}\conv\left\{-x_t,x_t\right\}
\end{equation}
where $\conv$ refers to the convex hull. Here, the addition symbol is meant in the Minkowski sum sense. 
It follows that $0\in \partial\left\|\phi(\theta^\star)\right\|_1$ is equivalent to the existence of a set of numbers $\lambda_t$ in $[-1,\; 1]$, $t\in \mathbb{I}^0(\theta^\star)$, such that \eqref{eq:iff-Cond} holds. 

\paragraph{Proof of S2}
Define two functions $q,h :\Re^n\rightarrow \Re_{\geq 0}$ by $q(\theta)=\sum_{t\notin \mathbb{I}^0(\theta^\star)}\left|y_t-\theta^\top x_t\right|$ and  $h(\theta)=\sum_{t\in \mathbb{I}^0(\theta^\star)}\left|y_t-\theta^\top x_t\right|$.
Then $\left\|\phi(\theta)\right\|_1=q(\theta)+h(\theta)$ and $q$ is differentiable at $\theta^\star$. It follows that  $\partial\left\|\phi(\theta^\star)\right\|_1=\nabla q(\theta^\star)+\partial h(\theta^\star)$, 
where $\nabla q(\theta^\star)$ is the gradient of $q$ at $\theta^\star$. We can hence write 
$$ \begin{aligned}
	\theta^\star \mbox{ minimizes } \left\|\phi(\theta)\right\|_1 
	&\Leftrightarrow \: 0\in \partial\left\|\phi(\theta^\star)\right\|_1  \\
	&\Leftrightarrow \: -\nabla q(\theta^\star) \in \partial h(\theta^\star).
\end{aligned}
$$ 
Note from \eqref{eq:Subdifferential} that $\partial h(\theta^\star)=\sum_{t\in \mathbb{I}^0(\theta^\star)}\conv\left\{-x_t,x_t\right\}$ so that $-\nabla q(\theta^\star) \in \partial h(\theta^\star)$ if and only if 
$\pm \nabla q(\theta^\star) \in \partial h(\theta^\star)$ and this in turn is equivalent to
$g^\top (\theta-\theta^\star) \leq h(\theta)-h(\theta^\star)$ $\forall \: \theta$, 
for $g\in \left\{-\nabla q(\theta^\star),+\nabla q(\theta^\star)\right\}$. 
It follows that $\theta^\star$ minimizes $\left\|\phi(\theta)\right\|_1$ if and only if 
\begin{equation}\label{eq:subgradient-inequality}
	\left|\nabla q(\theta^\star)^\top (\theta-\theta^\star) \right|\leq h(\theta)-h(\theta^\star)=\sum_{t\in \mathbb{I}^0(\theta^\star)}\left|(\theta-\theta^\star)^\top x_t\right| 
\end{equation}
for all $\theta$. 
The last equality is obtained by using the fact that $y_t-x_t^\top\theta^\star=0$ for all $t$ in $\mathbb{I}^0(\theta^\star)$. 
\noindent Finally the result follows by setting $\eta=\theta-\theta^\star$ and noting that $\nabla q(\theta^\star)=\sum_{t\in \mathbb{I}^+(\theta^\star)}x_t-\sum_{t\in \mathbb{I}^-(\theta^\star)}x_t$.

\paragraph{S1 $\Leftrightarrow$ S3}
The proof of the last equivalence is immediate. 
 
\paragraph{Uniqueness}
For convenience, we first prove S2'. Along the lines of the proof of S2 (see Eq. \eqref{eq:subgradient-inequality} and preceding arguments), we can see  that strict inequality in \eqref{eq:iff-Cond2} is equivalent to the following strict inequality
$-\nabla q(\theta^\star)^\top (\theta-\theta^\star)< h(\theta)-h(\theta^\star) \quad \forall \: \theta \neq \theta^\star$. 
On the other hand, $\nabla q(\theta^\star)^\top (\theta-\theta^\star)\leq  q(\theta)-q(\theta^\star) \quad \forall \: \theta$.  
Summing the two yields 
$$\left\|\phi(\theta^\star)\right\|_1= q(\theta^\star)+h(\theta^\star)<q(\theta)+h(\theta)=\left\|\phi(\theta)\right\|_1\: \forall \: \theta \neq \theta^\star.$$
Hence S2' is proved. \\
For the proof of S1', we proceed in two steps.\\
\textit{Sufficiency.} Assume $\rank(X_S)=n$. Then for any nonzero vector $\eta\in \Re^n$ there is at least one  $t_0\in S$ such that $\eta^\top x_{t_0}\neq 0$. Recall that $\left|\lambda_{t_0}\right|<1$ by definition of $S$. It follows that by multiplying  \eqref{eq:iff-Cond} on the left by $\eta^\top$ with $\eta\in \Re^n$ an arbitrary nonzero vector, and taking the absolute value yields \eqref{eq:iff-Cond2} with strict inequality symbol. We can therefore apply the proof of S2' to conclude that $\theta^\star$ is unique. \\ 
\textit{Necessity.}
Assume $\rank(X_S)<n$. Then pick any nonzero vector $\eta$ in $\ker(X_S^\top)$. Set $\eta_1=\nu \eta$ with $\nu\neq 0$. Indeed $\nu$ can be chosen sufficiently small such that $x_t^\top (\theta^\star+\eta_1)-y_t$ has the same sign as $x_t^\top \theta^\star-y_t$ for $t\in\mathbb{I}^-(\theta^\star)\cup \mathbb{I}^+(\theta^\star)$. For such values of $\nu$ we have $\mathbb{I}^+(\theta^\star)\subset\mathbb{I}^+(\theta^\star+\eta_1)$ and $\mathbb{I}^-(\theta^\star)\subset\mathbb{I}^-(\theta^\star+\eta_1)$. Moreover, since $\eta_1\in \ker( X_S^\top)$, $x_t^\top (\theta^\star+\eta_1)-y_t=\eta_1^\top x_t=0$ $\forall\: t \in S$, so that $S\subset \mathbb{I}^0(\theta^\star+\eta_1)$. Finally, it remains to re-assign the indices $t$ contained in $\mathbb{I}^0(\theta^\star)\setminus S$ for which $\lambda_t=1$. We get the following partition $\mathbb{I}^+(\theta^\star+\eta_1)=\mathbb{I}^+(\theta^\star)\cup \left\{t\in \mathbb{I}^0(\theta^\star):\eta_1^\top x_t>0\right\}$, $\mathbb{I}^-(\theta^\star+\eta_1)=\mathbb{I}^-(\theta^\star)\cup \left\{t\in \mathbb{I}^0(\theta^\star):\eta_1^\top x_t<0\right\}$, $\mathbb{I}^0(\theta^\star+\eta_1)=S\cup \left\{t\in \mathbb{I}^0(\theta^\star):\eta_1^\top x_t=0\right\}$.    
It follows that $\theta^\star+\eta_1\neq \theta^\star$ also satisfies \eqref{eq:iff-Cond} with the  sequence $\left\{\lambda_t\right\}_{t\in S}$ and is therefore a minimizer.  In conclusion, if $\rank(X_S)<n$, the minimizer cannot be unique. \qed
\end{pf}
\vspace*{-\baselineskip}
A number of important comments follow from Theorem \ref{thm:equivalence}. 
\begin{itemize}
	\item One first consequence of the theorem is that $\theta^o$
          can be computed exactly from a finite set of erroneous data (by solving problem \eqref{eq:L1-Cost})
          provided it satisfies the conditions S1' or S2' of the theorem. Note
          that there is no explicit boundedness condition imposed on the error
          sequence $\left\{f_t\right\}$. Hence the nonzero errors in
          this sequence can have arbitrarily large magnitudes as long as the optimization problem makes sense, \ie  provided $\|\phi(\theta^\star)\|_1$ remains finite.   
\item Second, the true parameter vector $\theta^o$ can be exactly recovered in the presence of, say, infinitely many nonzero errors $f_t$ (see also Proposition \ref{prop:infinite-errors}). For example, if the regressors $\left\{x_t\right\}$ satisfy 
$$\sum_{t\in \mathbb{I}^+(\theta^o)} x_t-\sum_{t\in \mathbb{I}^-(\theta^o)}x_t=0,  $$
and $\rank(X_{\mathbb{I}^0(\theta^o)})=n$, then by condition S2' $\theta^o$ is the unique solution to problem \eqref{eq:L1-Cost} regardless of the number of errors affecting the data. 
\item Third, if problem \eqref{eq:L1-Cost} admits a solution $\theta^\star$ that satisfies $y_t-x_t^\top\theta^\star\neq 0$  for all $t=1,\ldots,N$, then $\theta^\star$ is  non-unique. In effect,  $\mathbb{I}^0(\theta^\star)=\emptyset$ in this case and so, $\rank(X_{\mathbb{I}^0(\theta^\star)})=0<n$ which, by Theorem \ref{thm:equivalence},  implies non-uniqueness.  Indeed this is typically the case whenever the noise $\left\{e_t\right\}$ is nonzero. 
\end{itemize}
Another immediate consequence of Theorem \ref{thm:equivalence} can be stated as follows. 
\begin{cor}[On the special case of affine model]
If the model \eqref{eq:model} is affine in the sense that the regressor $x_t$ has the form $x_t= [\tilde{x}_t^\top\; 1]^\top$, with $\tilde{x}_t\in \Re^{n-1}$, then a necessary condition for $\theta^\star$ to solve problem \eqref{eq:L1-Cost} is that 
\begin{equation}\label{eq:necessary-affine}
	\Big|\left|\mathbb{I}^+(\theta^\star)\right|-\left|\mathbb{I}^-(\theta^\star)\right|\Big|\leq \left|\mathbb{I}^0(\theta^\star)\right|.
\end{equation}
Here, the outer bars $|\cdot|$ refer to absolute value while the inner ones which apply to sets refer to cardinality. 
\end{cor}
\vspace*{-\baselineskip}
\begin{pf}
The proof is immediate by considering the condition \eqref{eq:iff-Cond2} and taking $\eta=[\bm{0}^\top\; 1]^\top\in \Re^n$. \qed
\end{pf}
Eq. \eqref{eq:necessary-affine} implies that if the measurement model is affine and all the $f_t$'s have the same sign, \ie if one of the cardinalities $\left|\mathbb{I}^+(\theta^o)\right|$ or $\left|\mathbb{I}^-(\theta^o)\right|$ is equal to zero, then problem \eqref{eq:L1-Cost} cannot find the true $\theta^o$ whenever  more than $50\%$ of the elements of the sequence $\left\{f_t\right\}$  are nonzero. \\
Next, we discuss a special case in which the true parameter vector $\theta^o$ in \eqref{eq:model} can, in principle, be obtained asymptotically in the presence of an infinite number of nonzero errors $f_t$'s. 
\begin{prop}[Infinite number of outliers]
\label{prop:infinite-errors}
Assume that the error sequence $\left\{e_t\right\}$ in \eqref{eq:model} is identically  equal to zero. Assume further that the data $\left\{(x_t,y_t)\right\}_{t=1}^N$ are generated such that: 
\begin{itemize}
	\item There is a set $I^0\subset \mathbb{I}$ with  $\left|I^0\right|\geq n$, such that for any $t\in I^0$, $f_t=0$ and  $\rank(X_{I^0})=n$,
	\item For any $t\notin I^0$, $f_t$ is sampled from  a  distribution which is symmetric around zero.  
	\item The regression vector sequence $\left\{x_t\right\}\subset \Re^n$ is drawn from a probability distribution having a finite first moment. 
\end{itemize}
Then 
\begin{equation}\label{eq:Lim-Solution}
	\lim_{N\rightarrow \infty} \argmin_{\theta\in \Re^n}\dfrac{1}{N}\sum_{t=1}^N\left|y_t-x_t^\top \theta\right| = \left\{\theta^o\right\}
\end{equation}
with probability one. 
\end{prop}
\vspace*{-\baselineskip}
\begin{pf}
Under the conditions of the proposition, we have $\Prob(y_t-x_t^\top \theta^o<0)=\Prob(y_t-x_t^\top \theta^o>0)=1/2$, where $\Prob$ denotes  probability measure. 
It follows that $\left|\mathbb{I}^+(\theta^o)\right|$ and $\left|\mathbb{I}^-(\theta^o)\right|$ go jointly to infinity as the total number of samples $N$ tends to infinity. Hence, the expressions $\frac{1}{\left|\mathbb{I}^+(\theta^o)\right|}\sum_{t\in \mathbb{I}^+(\theta^o)}x_t$ and $\frac{1}{\left|\mathbb{I}^-(\theta^o)\right|}\sum_{t\in \mathbb{I}^-(\theta^o)}x_t$ are both sample estimates for the expectation of the process $\left\{x_t\right\}$.  By the law of large numbers, as $N\rightarrow \infty$, the two quantities converge to the true expectation of the process  $\left\{x_t\right\}$ with probability one, so that 
$$\lim_{N\rightarrow \infty } \Big[\dfrac{1}{\left|\mathbb{I}^+(\theta^o)\right|}\sum_{t\in \mathbb{I}^+(\theta^o)}x_t-\dfrac{1}{\left|\mathbb{I}^-(\theta^o)\right|}\sum_{t\in \mathbb{I}^-(\theta^o)}x_t\Big]=0.$$
As a consequence, $\theta^o$ satisfies condition S1' of Theorem \ref{thm:equivalence}  asymptotically with $\lambda_t=0$ for any $t\in \mathbb{I}^0(\theta^o)=I^0$. Hence the solution of the $\ell_1$ minimization problem  tends to $\theta^o$ with probability one  as the number of samples approaches infinity. \qed   
\end{pf}

\subsection{Worst-case necessary and sufficient conditions}
The conditions  \eqref{eq:iff-Cond}-\eqref{eq:min_alpha<1} derived in Theorem \ref{thm:equivalence} characterize completely the solution to the $\ell_1$-problem. However such conditions depend on which data points $(x_t,y_t)$ are affected by the gross errors and on the sign of the $f_t's$. We wish now to find  necessary and sufficient conditions that depend solely on the number of gross errors (or, equivalently on the number of zero elements in the sequence $\left\{f_t\right\}$). 
\begin{cor}[Necessary and sufficient conditions] 
\label{cor:NS-Conds} 
Let $d$ be an integer. Then the following statements are equivalent:
{
\setstretch{.3}
\begin{enumerate}
\item[\emph{(i)}] 
\begin{equation}\label{eq:(i)}
\begin{aligned}
		\forall \: \theta\in \Re^n, \forall \bm{y}\in & \Re^N, \: |\mathbb{I}^0(\theta)|\geq d \\
		                                & \Rightarrow \: \theta\in \argmin_{w\in \Re^n}\left\|\phi(w)\right\|_1
\end{aligned}
\end{equation}
\item[\emph{(ii)}] 
\begin{equation}\label{eq:(ii)}
	\!\!\!\!\!\!\!\!\max_{\substack{(I,I^c):\\ \left|I\right|=d}}\quad \!\!\!\!\! \max_{\eta\in \Re^n}\Big\{\left\|X_{I^c}^\top\eta\right\|_1 \suchthat \left\|X^\top \eta\right\|_1=1\Big\}\leq 1/2
\end{equation}
\item[\emph{(iii)}]
\begin{equation}
\label{eq:(iii)}
	\!\!\!\!\!\!\!\! \!\!\!\!\max_{\substack{(I,I^c):\\ \left|I\right|=d}}\: \: \max_{h\in \left\{\pm 1\right\}^{\left|I^c\right|}}\min_{\alpha\in \Re^{|I|}}\Big\{\left\|\alpha\right\|_\infty \suchthat  X_{I^c}h=X_I\alpha\Big\}\leq 1
\end{equation}
\end{enumerate}
}
\end{cor}
In \eqref{eq:(ii)}-\eqref{eq:(iii)} and similar equations in the paper, the leftmost maximum  is taken over the set of those partitions $(I,I^c)$ of $\mathbb{I}$ that satisfy $\left|I\right|=d$. Eq. \eqref{eq:(iii)} should be read with the implicit assumption that the inequality fails to hold whenever the optimization problem is not feasible. 
\begin{pf}[of Corollary \ref{cor:NS-Conds}]
That (ii) and (iii) are equivalent is a statement that results directly from the equivalence of \eqref{eq:iff-Cond2} and \eqref{eq:min_alpha<1} in Theorem \ref{thm:equivalence}. To see this, let $\theta\in\Re^n$ be a solution to problem \eqref{eq:L1-Cost} and set 
 $I^c = \mathbb{I}^-(\theta)\cup\mathbb{I}^+(\theta)$, $I=\mathbb{I}^0(\theta)$, $h_{I^c}\in \left\{-1,+1\right\}^{|I^c|}$ such that $h_i=+1$   if $i\in \mathbb{I}^+(\theta)$ and $h_i=-1$ if $i\in \mathbb{I}^-(\theta)$.  Then Eq. \eqref{eq:iff-Cond2} can be  written as
 \begin{equation}\label{eq:Equiv(ii)}
	 \begin{aligned}
		 &\left|\eta^\top X_{I^c}h_{I^c}\right|\leq\left\|X_I^\top \eta\right\|_1 \quad \forall \eta \in \Re^n \\
		 \Leftrightarrow & \max_{\eta\in \Re^n}\Big\{\left|\eta^\top X_{I^c}h_{I^c}\right|+\left\|X_{I^c}^\top\eta\right\|_1 \suchthat \left\|X^\top \eta\right\|_1=1\Big\}\leq 1.
	 \end{aligned}
 \end{equation}
Similarly Eq. \eqref{eq:min_alpha<1} reads as 
\begin{equation}\label{eq:Equiv(iii)}
	\min_{\alpha_I}\Big\{\left\|\alpha_I\right\|_\infty\suchthat X_{I^c}h_{I^c}=X_I\alpha_I\Big\}\leq 1.
\end{equation}
The equivalence  (ii) $\Leftrightarrow$ (iii) then follows by applying the chains of maximums  
$\max_{(I,I^c):\left|I\right|=d}\max_{h_{I^c}\in \left\{\pm 1\right\}^{\left|I^c\right|}}$
to each of the equations \eqref{eq:Equiv(ii)} and \eqref{eq:Equiv(iii)} and noting that $\max_{h_{I^c}\in \left\{\pm 1\right\}^{\left|I^c\right|}} \left|\eta^\top X_{I^c}h_{I^c}\right|=\|X_{I^c}^\top \eta\|_1$. 

We shall now establish the equivalence  (i) $\Leftrightarrow$ (ii). 
Let $\theta\in \Re^n$ and $\bm{y}\in \Re^N$ be any vectors  such that $I=\left|\mathbb{I}^0(\theta)\right|=d$. The so-defined $I$ can be any subset of $\mathbb{I}$ provided $\left|I\right|=d$. Hence any $\theta$ satisfying this cardinality constraint solves problem \eqref{eq:L1-Cost} if and only if 
\eqref{eq:Equiv(ii)} holds for any partition $(I,I^c)$ of $\mathbb{I}$ with $\left|I\right|=d$ and for any $h_{I^c}\in \left\{- 1,+1\right\}^{|I^c|}$. This is equivalent to Eq. \eqref{eq:(ii)}. 

Finally, let us observe that 
	$$\max_{(I,I^c):\left|I\right|=d}\quad \!\!\!\!\! \max_{\eta\in \Re^n}\Big\{\left\|X_{I^c}^\top\eta\right\|_1 \suchthat \left\|X^\top \eta\right\|_1=1\Big\}$$
	is a decreasing function of $d$ so that if \eqref{eq:(ii)} holds for some $d_0$, it holds also for any $d\geq d_0$. It follows that (i) $\Leftrightarrow$ (ii), hence completing the proof. \qed   
\end{pf}
\vspace*{-\baselineskip}
It should be mentioned that the equivalence (i) $\Leftrightarrow$ (ii) was also obtained in earlier papers, see e.g., \cite{Xu14-Automatica,Xu11-IT}. Uniqueness of the solution follow in a similar way as in the proof of Corollary \ref{cor:NS-Conds} by invoking conditions S1' and S2' of Theorem \ref{thm:equivalence}. 
\begin{cor}[Uniqueness] 
\label{cor:NS-Conds-2}
Let $d$ be an integer. Then the following statements are equivalent.
{\setstretch{.2} 
\begin{enumerate}
\itemsep=.05cm
\item[\emph{(i')}]   
\begin{equation}\label{eq:(i')}
\begin{aligned}
\forall \: \theta\in \Re^n, \forall \bm{y}\in & \Re^N, \: |\mathbb{I}^0(\theta)|\geq d \\
	                     & \Rightarrow \: \argmin_{w\in \Re^n}\left\|\phi(w)\right\|_1=\big\{\theta\big\}
\end{aligned}
\end{equation}
\item[\emph{(ii')}] Eq. \eqref{eq:(ii)} holds with strict inequality. 
\item[\emph{(iii')}]
\begin{equation}
\label{eq:(iii')}
\renewcommand\arraystretch{.5}
\begin{aligned}
		\!\!\!\!\!\!\!\! \!\!\!\!\max_{\substack{(I,I^c):\\\left|I\right|=d}}\:  &\max_{h\in \left\{\pm 1\right\}^{\left|I^c\right|}}\min_{\alpha}\Big\{\left\|\alpha\right\|_\infty  \suchthat\Bigg.\\
	         &\: \:  \: X_{I^c}h=X_I\alpha,\\
	         & \: \: \:  \exists S\subset I, \rank(X_S)=n, \left\|\alpha_S\right\|_\infty<1\Big\}\leq 1 													
\end{aligned}
\end{equation}
\end{enumerate}
}
\end{cor}
\begin{rem}
Corollaries \ref{cor:NS-Conds} and \ref{cor:NS-Conds-2} imply the following. If there exists an integer $d$ such that \eqref{eq:(ii)} or \eqref{eq:(iii)} holds and  $\left\{\theta\in \Re^n:\left\|\phi(\theta)\right\|_0\leq N-d\right\}\neq \emptyset$, then $\argmin_{\theta}\left\|\phi(\theta)\right\|_0\subset \argmin_{\theta}\left\|\phi(\theta)\right\|_1$. 
It follows under these conditions that whenever $\theta^o$ solves the $\ell_0$ problem \eqref{eq:L0-Pbm}, it also solves the $\ell_1$ problem \eqref{eq:L1-Cost}. 
In particular $\argmin_{\theta}\left\|\phi(\theta)\right\|_1=\left\{\theta^o\right\}$ $\Rightarrow$ 
$\argmin_{\theta}\left\|\phi(\theta)\right\|_0=\left\{\theta^o\right\}$.  
\end{rem}

It should be noted that when the data are noise-free, there always exists a $d$ such that \eqref{eq:(i)}-\eqref{eq:(iii)} hold. For example $d=N$ is the maximum possible value that satisfies these conditions. 
Let us denote by $\pi^o(X)$ the minimum integer $d$ such that the conditions \eqref{eq:(i)}-\eqref{eq:(iii)} hold, that is, 
\begin{equation}\label{eq:Pi-O}
	\pi^o(X)=\min\big\{d\in \mathbb{I}\: \suchthat  \mbox{ Eq. \eqref{eq:(ii)} is true}\big\}.
\end{equation}
Such a number $\pi^o(X)$ depends only on the matrix $X$. It  can be viewed as a measure of the richness properties of the regressor matrix $X$. Recoverability of the true parameter vector $\theta^o$ by the least $\ell_1$-norm estimator \eqref{eq:L1-Cost} in the face of gross errors is enhanced when $\pi^o(X)$ is small. We may hence say that the smaller  $\pi^o(X)$, the richer (or more generic) the regressors in $X$ are.
 
Computing $\pi^o(X)$ directly from the definition \eqref{eq:Pi-O} is a hard combinatorial problem with a complexity comparable to that of the $\ell_0$ problem \eqref{eq:L0-Pbm}. An algorithm of slightly reduced complexity but still combinatorial has been derived in \cite{Sharon09-ACC} for this purpose.  Here, we ask the question of whether $\pi^o(X)$ can be more cheaply estimated in a somewhat efficient way. Such estimates are most likely over-estimates and lead to sufficient conditions for exact recoverability of the parameter vector $\theta^o$ in the presence of gross errors sequence $\left\{f_t\right\}$. 
\vspace*{-\baselineskip}
\subsection{Sufficient conditions of recoverability by convex optimization}
\label{subsec:Sufficient-Conditions}
We start by introducing the following notations : 
\begin{align}
	&v_1(k)=\max_{(I,I^c):\left|I\right|=k\geq \nu_n(X)}\left\|X_I^\top(X_IX_I^\top)^{-1}X_{I^c}\right\|_\infty \label{eq:v1(k)}\\
	&v_2(k)=\max_{(I,I^c):\left|I\right|=k}\left\|X_{I^c}^\top(XX^\top)^{-1}X\right\|_1, \label{eq:v2(k)}
\end{align}
where the maximum is taken over the set of those partitions $(I,I^c)$ of $\mathbb{I}$ that satisfy $\left|I\right|=k$. 
In addition, let 
\begin{equation}\label{eq:k1(X)}
	k_1(X)=\min_{k\in \mathbb{I},k\geq \nu_n(X)}\big\{k: v_1(k)\leq 1\big\}
\end{equation}
 and 
\begin{equation}\label{eq:k2(X)}
	k_2(X)=\min_{k\in \mathbb{I}}\big\{k: v_2(k)\leq 1/2\big\}.
\end{equation}
Assuming that $\rank(X)=n$, it can be seen that the numbers  $k_i(X)$, $i=1,2$,  are well-defined. 
First, note that $\nu_n(X)\leq N$  so that the condition $k_i(X)\geq \nu_n(X)$ is achievable. Moreover, by considering the trivial partition $(I,I^c)$ with $I=\mathbb{I}$ and $I^c=\emptyset$, we see that a possible (the largest indeed) value for $k_i(X)$ is $N$.

\begin{thm}[Sufficient condition for exact recovery]\label{thm:X'(XX')}
Assume $\rank(X)=n$. Then the numbers $k_1(X)$ and $k_2(X)$ satisfy 
	\begin{multline}\label{eq:min-k1-k2}
		\forall \theta \in \Re^n, \forall \bm{y}\in  \Re^N \: \left|\mathbb{I}^0(\theta)\right|\geq \min(k_1(X),k_2(X)) \\\Rightarrow \theta\in \argmin_{w\in \Re^n}\left\|\phi(w)\right\|_1
	\end{multline}
where $\phi(w)=\bm{y}-X^\top w$. 	
If in addition all the inequalities involved in the definition of $k_1(X)$ and $k_2(X)$ are strict, then the second part of \eqref{eq:min-k1-k2}  becomes $\argmin_{w\in \Re^n}\left\|\phi(w)\right\|_1=\left\{\theta\right\}$, that is, $\theta$ is the unique minimizer of \eqref{eq:L1-Cost}. 
\end{thm}
\vspace*{-\baselineskip}
\begin{pf}
To prove the first statement, we just need to show that 
\begin{equation}
\min\big(k_1(X),k_2(X)\big)\geq \pi^o(X). 
\end{equation}
\textit{Part 1:} $k_1(X)\geq \pi^o(X)$. \\ 
Define  
$$v_0(k)=\max_{\substack{(I,I^c):\\ \left|I\right|=k}}\: \: \max_{h\in \left\{\pm 1\right\}^{\left|I^c\right|}}\min_{\alpha\in \Re^{|I|}}\Big\{\left\|\alpha\right\|_\infty \suchthat  X_{I^c}h=X_I\alpha\Big\}
$$
that is, $v_0(k)$ corresponds to the left hand side of \eqref{eq:(iii)} (with $d$ replaced by $k$). 
By making use of Corollary \ref{cor:NS-Conds} and the definitions \eqref{eq:v1(k)} and \eqref{eq:k1(X)}, it is enough to show that $v_0(k)\leq v_1(k)$. For this purpose, let $(I,I^c)$ be an arbitrary partition of $\mathbb{I}$ such that $|I|\geq \nu_n(X)$. 
Consider the problem
\begin{equation}\label{eq:alpha-opt}
	\min_{\alpha\in \Re^{|I|}}\big\{\left\|\alpha\right\|_\infty : X_{I^c}h=X_{I}\alpha\big\},
\end{equation}
where  $h\in \left\{-1,+1\right\}^{|I^c|}$ but otherwise arbitrary. 
 Let $p^*(h)$ be the optimal value  of problem  \eqref{eq:alpha-opt}  and pose 
$$\alpha^*(h) = \argmin_{\alpha\in \Re^{|I|}}\big\{\left\|\alpha\right\|_2 \mbox{ s.t. } X_{I^c}h=X_{I}\alpha\big\}.$$
 Since $\alpha^*(h)$ is a feasible point for problem \eqref{eq:alpha-opt},  it must hold  that   $p^*(h)\leq \left\|\alpha^*(h)\right\|_\infty$. 
 The so-defined   $\alpha^*(h)$ is the well-known least Euclidean-norm solution to an underdetermined system of linear equations \cite{Boyd04-Book}; $\alpha^*(h)$ can be analytically expressed as $\alpha^*(h) =X_I^\top(X_IX_I^\top)^{-1}X_{I^c}h$ for all $h\in \left\{-1,+1\right\}^{|I^c|}$.
As a consequence, 
$$
\begin{aligned}
	\max_{h\in \left\{\pm 1\right\}^{\left|I^c\right|}} p^*(h)&\leq \max_{h\in \left\{\pm 1\right\}^{\left|I^c\right|}}\left\|X_I^\top(X_IX_I^\top)^{-1}X_{I^c}h\right\|_\infty \\
	&\leq \max_{h\in \left\{\pm 1\right\}^{\left|I^c\right|}}\left\|X_I^\top(X_IX_I^\top)^{-1}X_{I^c}\right\|_\infty \left\|h\right\|_\infty\\
	& =\left\|X_I^\top(X_IX_I^\top)^{-1}X_{I^c}\right\|_\infty. 
\end{aligned}
$$ 
 The last equality uses  $\left\|h\right\|_\infty =1$.   
It follows that if  $v_0(k)\leq v_1(k)$ hence proving that $k_1(X)\geq \pi^o(X)$.

\textit{Part 2:} $k_2(X)\geq \pi^o(X)$ \\
Proceeding from Corollary \ref{cor:NS-Conds} and the definitions \eqref{eq:v2(k)} and \eqref{eq:k2(X)},  we just need to show that 
$$
\max_{(I,I^c):\left|I\right|=k}\quad \!\!\!\!\! \max_{\eta\in \Re^n}\Big\{\left\|X_{I^c}^\top\eta\right\|_1 \suchthat \left\|X^\top \eta\right\|_1=1\Big\}\leq v_2(k).
$$ 
To this end, set  $b=X^\top \eta$. Then $b\in \im(X^\top)$ and $\eta=(XX^\top)^{-1}Xb$. It follows that  
$$
\begin{aligned}
	\max_{\eta\in \Re^n}&\Big\{\left\|X_{I^c}^\top\eta\right\|_1 \suchthat \left\|X^\top \eta\right\|_1=1\Big\}\\
	& \!\!\!= \max_{b\in \im(X^\top)}\Big\{\left\|X_{I^c}^\top (XX^\top)^{-1}Xb\right\|_1 \suchthat \left\|b\right\|_1=1\Big\}\\
	&\!\!\!\leq \left\|X_{I^c}^\top (XX^\top)^{-1}X\right\|_1.
\end{aligned}
$$ 
Taking now the maximum over all partitions $(I,I^c)$ of $\mathbb{I}$, $|I|=d$, the result follows. 

\textit{Uniqueness.}
This is a straightforward consequence of Corollary \ref{cor:NS-Conds-2}. \qed
\end{pf}
\vspace*{-\baselineskip}
Evaluating numerically  $k_1(X)$ and $k_2(X)$ is still a combinatorial problem. Next we investigate some over-estimates of $\pi^o(X)$ which are free from the maximization over sets $(I,I^c)$. The new thresholds have the important advantage of being more easily computable.  
\begin{thm}[Another sufficient condition]
\label{prop:sufficient-condition} 
Assume that $\nu_n(X)\leq N-1$ and define the following numbers  
\begin{align}
	&r(X)=\max_{k\in \mathbb{I}}\left|x_k(XX^\top)^{-1}x_k\right| \label{eq:r(X)} \\
	& \xi(X)=\max_{k\in \mathbb{I}}\min_{\gamma_k\in \Re^{N-1}}\Big\{\left\|\gamma_k\right\|_\infty \suchthat x_k=X_{\neq k}\gamma_k\Big\} \label{eq:xi(X)} 
\end{align}
where $X_{\neq k}\triangleq X_{\mathbb{I}\setminus \{k\}}$ is the matrix obtained from $X$ by removing its $k$-th column. 
Then the following statement is true: 
$\forall p\in \left\{\frac{1}{r(X)},1+\frac{1}{\xi(X)}\right\} $,  
\begin{equation}\label{eq:sufficient}
\begin{aligned}
	\forall \theta\in \Re^n, \forall \bm{y}\in \Re^N,\: 	&	\left|\mathbb{I}^0(\theta)\right|>N-\dfrac{p}{2} \\
	&\Rightarrow \argmin_{w\in \Re^n}\left\|\phi(w)\right\|_1=\left\{\theta\right\}. 
\end{aligned}
\end{equation}
\end{thm}
\vspace*{-\baselineskip}
\begin{pf}
The proof is decomposed into two cases. \\
\textit{Case 1:} $p=1/r(X)$.   
	From Theorem \ref{thm:X'(XX')}, it is known that 
	$\left\|X_{I^c}^\top (XX^\top)^{-1}X\right\|_1<1/2$, $I^c=\mathbb{I}\setminus\mathbb{I}^0(\theta)$, is a sufficient condition for $\theta$ to be the unique minimizer of \eqref{eq:L1-Cost}.
Now we use the fact that the $1$-norm of a matrix is the maximum of the $1$-norms of its columns:  
$$
	\begin{aligned}
		\left\|X_{I^c}^\top (XX^\top)^{-1}X\right\|_1& =\max_{t=1,\ldots,N} \left\|X_{I^c}^\top (XX^\top)^{-1}x_t\right\|_1\\
		& = \max_{t=1,\ldots,N} \sum_{k\in I^c}\left|x_k^\top (XX^\top)^{-1}x_t\right|\\
		& \leq \left|I^c\right| r(X).
	\end{aligned}
$$
Therefore a sufficient condition for $\theta$ to be the unique solution of \eqref{eq:L1-Cost} is that
$\left|I^c\right| r(X)<1/2$. The conclusion follows immediately.

\textit{Case 2:} $p=1+1/\xi(X)$. \\  
Since $\nu_n(X)\leq N-1$, each $x_k$, $k\in \mathbb{I}$, can be written as a linear combination of the columns of $X_{\neq k}$. Let $\gamma_k\in \Re^{N-1}$ be any vector satisfying  
$x_k= X_{\neq k}\gamma_k$. It follows that for any $\eta\in \Re^n$,
$$\left|x_k^\top \eta\right|  
\leq \sum_{t\neq k}\left|\gamma_{k,t}\right| \left|x_t^\top \eta\right|\leq \left\|\gamma_k\right\|_\infty\left(\left\|X^\top\eta\right\|_1-\left|x_k^\top \eta\right|\right)$$
with $\gamma_{k,t}$ denoting the entry of $\gamma_k\in \Re^{N-1}$ indexed by $t$. 
Since this holds for any $\gamma_k$ such that $x_k= X_{\neq k}\gamma_k$, it holds also for 
$$\gamma_k^\star=\argmin_{\gamma\in \Re^{N-1}}\Big\{\left\|\gamma\right\|_\infty \suchthat x_k=X_{\neq k}\gamma\Big\}.$$
Hence, 
$$\left|x_k^\top \eta\right| \leq \xi(X)\left(\left\|X^\top\eta\right\|_1-\left|x_k^\top \eta\right|\right)\quad \forall k\in \mathbb{I}, \forall \eta\in \Re^n$$
or, equivalently,
$$\left|x_k^\top \eta\right| \leq \dfrac{\xi(X)} {1+\xi(X)}\left\|X^\top\eta\right\|_1\quad \forall k\in \mathbb{I}, \: \forall \eta\in \Re^n.$$
Summing over the set $I^c$ yields
\begin{equation}\label{eq:concentration-ratio}
	\max_{\eta\neq 0}\dfrac{\left\|X_{I^c}^\top\eta\right\|_1}{\left\|X^\top\eta\right\|_1}\leq \dfrac{\xi(X)} {1+\xi(X)}\left|I^c\right|
\end{equation}
In virtue of \eqref{eq:(ii)}, it appears that for $\theta$ to be the unique minimizer of \eqref{eq:L1-Cost}, it is sufficient that 
$\dfrac{\xi(X)} {1+\xi(X)}\left|I^c\right|<1/2$ from which we see that $\left|I^c\right|<1/2(1+1/\xi(X))$ is a sufficient condition. \qed 
\end{pf}
\vspace*{-\baselineskip}
It should be noted that the numbers $r(X)$ and $\xi(X)$ defined in \eqref{eq:r(X)}  and \eqref{eq:xi(X)}  are both computable from matrix $X$. $r(X)$ is less expensive to evaluate numerically  than $\xi(X)$ but leads in general to a more pessimistic bound than $\xi(X)$ on the number of tolerable outliers. Computing $\xi(X)$ literally from the definition \eqref{eq:xi(X)}, for example by interior point methods,  requires solving about $N$ linear programs having each a worst-case complexity bounded by $O(\sqrt{N}\ln(1/\epsilon))$ where $\epsilon$ refers to the precision demanded \cite{Gondzio12}. 
Empirical evidence tend to suggest that the bound obtained from $\xi(X)$ on the number of correctable outliers is very close to $N-\pi^o(X)$ (see Section \ref{subsec:bounds}). As it turns out, while the computational complexity (polynomial) of $\xi(X)$ is lower than that of the algorithm developed in \cite{Sharon09-ACC} for estimating directly $\pi^o(X)$, it still provides a competitive bound.  
\begin{rem}
$\xi(X)$ can be approximated at a cheaper computational cost by replacing the infinity norm with the $2$-norm. This provides an over-estimate $\widehat{\xi}(X)\geq \xi(X)$ defined by 
$$\begin{aligned}
	\widehat{\xi}(X)&=\max_{k\in \mathbb{I}}\min_{\gamma\in \Re^{N-1}}\Big\{\left\|\gamma\right\|_2 \: \suchthat x_k=X_{\neq k}\gamma\Big\} \label{eq:xi(X)-hat}\\
	&=\max_{k\in \mathbb{I}}\left\|X_{\neq k}^\top(X_{\neq k}X_{\neq k}^\top)^{-1}x_k\right\|_2.
\end{aligned}  
$$
\end{rem}
We conclude this section with a few technical remarks concerning some interesting properties of the numbers $r(X)$ and $\xi(X)$. 
\begin{lem}[Some properties of $r(X)$]\label{lem:r(X)}
Under the assumption that $\nu_n(X)\leq N-1$, $r(X)$ and $\xi(X)$ satisfy:  
\begin{equation}\label{eq:r(X)-nu(X)}
	\max\Big\{\frac{1}{r(X)},1+\frac{1}{\xi(X)}\Big\}\leq N-\nu_n(X)+1,
\end{equation} 
\begin{equation}\label{eq:r(X)-k2(X)}
	N-\dfrac{1}{2r(X)}\geq k_2(X)\geq \pi^o(X).
\end{equation}
\end{lem}
\vspace*{-\baselineskip}
\begin{pf}
\textit{Proof of \eqref{eq:r(X)-nu(X)}: }\\
\textit{First case: $1/r(X)\leq N-\nu_n(X)+1$} . 
We know from the proof of Theorem \ref{prop:sufficient-condition} (see also Part 2 in the proof of Theorem \ref{thm:X'(XX')}) that 
$$ r(X)\geq \dfrac{1}{\left|I\right|}\left\|X_{I}^\top (XX^\top)^{-1}X\right\|_1\geq\dfrac{1}{\left|I\right|}\max_{\eta \neq 0}\dfrac{\left\|X_{I}^\top \eta\right\|_1}{\left\|X^\top \eta\right\|_1}$$
for any $I\subset \mathbb{I}$. A special case is when the subset $I$ is a singleton of the form $I=\left\{q\right\}$. 
For any $\eta\in \Re^n$, let $T(\eta)=\supp\left(X^\top \eta\right)=\left\{t\in \mathbb{I}:x_t^\top \eta \neq 0\right\}$. 
When $\eta\neq 0$, consider an index $q(\eta)$ such that $q(\eta)\in \argmax_{k\in T(\eta)}\left|x_k^\top \eta\right|$. Then by applying the above inequality with $I=\left\{q(\eta)\right\}$, we get 
$$r(X)\geq \max_{\eta \neq 0}\dfrac{\big|x_{q(\eta)}^\top \eta\big|}{\big\|X_{T(\eta)}^\top \eta\big\|_1}\geq \max_{\eta \neq 0}\dfrac{1}{\big|T(\eta)\big|}  $$
with $\big|T(\eta) \big|$ standing for the cardinality of $T(\eta)$. When $\eta\neq 0$, the smallest value $\big|T(\eta) \big|$ can take is $N-\nu_n(X)+1$ where $\nu_n(X)$ is the number defined by Eq. \eqref{eq:Nu_n_X}. It can therefore be concluded that $r(X)\geq 1/(N-\nu_n(X)+1)$. \\
\textit{Second case: $1+1/\xi(X)\leq N-\nu_n(X)+1$.}
The second case follows by a similar reasoning as in the first one. In effect, according to \cite{Cadzow73-SIAM}, the following equality holds, 
$$\begin{aligned}
	&\min_{\gamma_k\in \Re^{N-1}}\Big\{\left\|\gamma_k\right\|_\infty \suchthat x_k=X_{\neq k}\gamma_k\Big\} \\
	& = \max_{\eta\in \Re^n}\Big\{x_k^\top \eta \: \suchthat \big\|X_{\neq k}^\top\eta\big\|_1=1\Big\}\\
	& =\max_{\eta \neq 0}\Big\{\dfrac{|x_k^\top \eta|}{\big\|X_{\neq k}^\top \eta\big\|_1}\: \suchthat\: x_k^\top\eta\geq 0\Big\}. 
\end{aligned} 
$$
For a given $\eta \neq 0$, pick $q_1(\eta)$ such that 
$$q_1(\eta)\in \argmax_{k\in \mathbb{I}}\left\{x_k^\top \eta: x_k^\top \eta\geq 0\right\}.$$ 
By exploiting the equalities above and using the notation  $T(\eta)$ defined earlier we get that 
$$\xi(X)\geq \max_{\eta \neq 0}\dfrac{|x_{q_1(\eta)}^\top \eta|}{\big\|X_{\neq k}^\top \eta\big\|_1}.$$
Now the conclusion can be reached by arguing similarly as in the \textit{first case}. 

\textit{Proof of \eqref{eq:r(X)-k2(X)}: }
Let $(I,I^c)$ be a partition of $\mathbb{I}$ and set $k=N-1/\left(2 r(X) \right)$. First note that 
$$\left|I\right|\geq k \: \Rightarrow \:  \left|I^c\right| r(X)\leq 1/2. $$ 
On the other hand, we know (from the proof of Theorem \ref{prop:sufficient-condition}) that
$$\big\|X_{I^c}^\top\left(XX^\top\right)^{-1}X\big\|_1\leq  \left|I^c\right| r(X).$$
It follows that $\left|I\right|\geq k$ $\Rightarrow$ $\big\|X_{I^c}^\top\left(XX^\top\right)^{-1}X\big\|_1\leq 1/2$ and hence  
$v_2(k)\leq 1/2$. By invoking the definition of the number $k_2(X)$ in \eqref{eq:k2(X)}, it  can be concluded that  $k\geq k_2(X)$. \qed 
\end{pf}
\vspace{-13pt}
\begin{rem}
For any  nonsingular matrix $T\in \Re^{n\times n}$, $r(TX)=r(X)$, $\xi(TX)=\xi(X)$, $k_i(TX)=k_i(X)$, $i=1,2$. It follows that the numbers $r(X)$, $\xi(X)$ and $k_i(X)$, $i=1,2$, 
depend only on the subspace spanned by the rows of the regressor matrix $X$. 
\end{rem}
\section{Some implementation aspects}
\subsection{Enforcing recoverability by iterative re-weighting}\label{subsec:Sparsity-Enhancing}
The parameter vector $\theta^o$ from the model \eqref{eq:model} can be uniquely recovered by solving the convex problem \eqref{eq:L1-Cost}  if  $\theta^o$ satisfies, for example,  condition \eqref{eq:sufficient} of Theorem \ref{prop:sufficient-condition}. In case this condition is not naturally satisfied, an interesting question is how we can process the data in order to promote it. In this section we discuss an algorithmic strategy for enhancing the recoverability of $\theta^o$ by $\ell_1$ minimization. Our discussion is inspired by \cite{Candes08-JFAA}.   
The idea is to solve a sequence of problems of the type \rf{eq:L1-Cost} with different weights computed iteratively \cite{Candes08-JFAA,Bako11-Automatica}. 
The iterative scheme can be defined for a fixed number $r_{\text{\tiny max}}$ of
iterations as follows. At  iteration $r=0,\ldots,r_{\text{\tiny max}}$, compute 
\begin{equation}\label{eq:theta(r)}
	\theta^{(r)}=\argmin_{\theta\in \Re^n}\sum_{t=1}^{N} w_t^{(r)}\left|y_t-\theta^\top x_t\right|, 
\end{equation}
with weights defined, for all $t$, by $w_t^{(0)}=1/N$, 
and
$$w_t^{(r)}= \dfrac{\xi_t^{(r)}}{\sum_{t=1}^{N} \xi_t^{(r)}}, \quad \mbox{if } r\geq 1,
$$
where $$\xi_t^{(r)}=\dfrac{1}{\left|y_t-x_t^\top\theta^{(r-1)}\right|+\delta},$$
with  $\delta>0$ a small number whose role is to prevent division by zero and $r$ is the iteration number. Note that there are many other reweighting strategies which can be used in \eqref{eq:theta(r)}, see \eg \cite{Chen10-TechReport,Zhao12,Le14-TAC}. 
Since we are dealing here with a sequence of convex optimization problems, they can be numerically implemented using any convex solver. In particular the \verb!CVX! software package \cite{Grant-Boyd_CVX} solves efficiently this category of problems in a Matlab environment. 

\subsection{On the treatment of the noise $\left\{e_t\right\}$}
\label{subsec:Noise}
The formulations \eqref{eq:L0-Pbm} and \eqref{eq:L1-Cost} are convenient when the noise $\left\{e_t\right\}$ is equal to zero. Nevertheless, they are expected to work in the presence of a moderate amount of noise. To take explicitly the noise $\left\{e_t\right\}$ into account, we propose to compute estimates $\hat{\bm{e}}\in \Re^N$ and $\varphi\in \Re^N$ (of $\bm{e}$ and $\bm{f}$ respectively) by minimizing a cost of the form $\|\hat{\bm{e}} \|_2^2+\lambda \|\varphi\|_0 $ under an equality constraint of the form \eqref{eq:model}. 
In other words, we consider the problem 
\begin{equation}
	\minimize_{(\theta,\varphi)\in \Re^n\times \Re^N}  \Big[\dfrac{1}{2}\big\|\bm{y}-X^\top \theta-\varphi \big\|_2^2+\lambda \big\|\varphi\big\|_0\Big]  
\end{equation}
and its convex relaxation,
\begin{equation}\label{eq:regularized-L1}
	\minimize_{(\theta,\varphi)\in \Re^n\times \Re^N} \Big[\dfrac{1}{2}\big\|\bm{y}-X^\top \theta-\varphi \big\|_2^2+\lambda \big\|\varphi\big\|_1\Big]. 
\end{equation}
where $\lambda\geq 0$ is a regularization parameter. 
\begin{lem}\label{lem1:opt-conditions-regularized}
A pair  $(\theta^\star,\varphi^\star)\in \Re^n\times \Re^N$  solves \eqref{eq:regularized-L1} if and only if it satisfies 
\begin{align}
&	XX^\top \theta^\star-X\big(\bm{y}-\varphi^\star\big)=0 \label{eq:opt-cond1-regularized-L1}\\
&	X^\top \theta^\star-\big(\bm{y}-\varphi^\star\big)=-\lambda s(\varphi^\star), \label{eq:opt-cond2-regularized-L1}
\end{align}
where $s(\varphi^\star)$ is a vector in $\Re^N$  whose entries $s_t(\varphi^\star)$, $t=1,\ldots,N$, are defined by:  $s_t(\varphi^\star)=\sign(\varphi_t^\star)$ if  $\varphi_t^\star\neq 0$ and $s_t(\varphi^\star)\in \left[-1,\: 1\right]$ if  $\varphi_t^\star= 0$. 
\end{lem}
\begin{pf}
Let $l(\theta,\varphi)=\dfrac{1}{2}\big\|\bm{y}-X^\top \theta-\varphi \big\|_2^2+\lambda \big\|\varphi\big\|_1$ be the objective function of the problem \eqref{eq:regularized-L1}. Then $l$ is a proper convex function which is differentiable with respect to variable $\theta$ on $\Re^n$ and admits a subdifferential at any variable $\varphi\in \Re^N$. $(\theta^\star,\varphi^\star)$ minimizes $l(\theta,\varphi)$ iff
$0=\nabla_\theta l(\theta^\star,\varphi^\star)$ and  $0\in \partial_{\varphi}l(\theta^\star,\varphi^\star)$. These conditions translate immediately into $XX^\top \theta^\star-X(\bm{y}-\varphi^\star)=0$ and $-\left(\bm{y}-\varphi^\star-X^\top \theta^\star\right)+\lambda s(\varphi^\star)=0$, where $s(\varphi^\star)\in \partial\left\|\varphi^\star\right\|_1$ is a subgradient of $\left\|\varphi\right\|_1$ at $\varphi^\star$. \qed
\end{pf}
It is interesting to note that \eqref{eq:opt-cond1-regularized-L1}-\eqref{eq:opt-cond2-regularized-L1} imply $Xs(\varphi^\star)=0$, which is very similar to \eqref{eq:iff-Cond}. The following lemma characterizes the uniqueness of the solution of \eqref{eq:regularized-L1}. 
\begin{lem}[Uniqueness of solution to \eqref{eq:regularized-L1}]\label{lem2:opt-conditions-regularized}
A pair $(\theta^\star,\varphi^\star)$ is the unique solution to problem \eqref{eq:regularized-L1}   if and only if both of the following statements are true 
\begin{itemize}
	\item[(i)] $(\theta^\star,\varphi^\star)$ satisfies conditions \eqref{eq:opt-cond1-regularized-L1}-\eqref{eq:opt-cond2-regularized-L1} for some $s(\varphi^\star)\in \partial\left\|\varphi^\star\right\|_1$
	\item[(ii)] $\rank(X)=n$ and $\rank(\Psi_{\mathcal{S}^c})=\left|\mathcal{S}^c\right|$.\\
	Here, $\Psi = I_N-X^\top(XX^\top)^{-1} X$, with $I_N$ being the identity matrix of order $N$, $\Psi_{\mathcal{S}^c}$ is a matrix formed with the columns of $\Psi$ indexed by $\mathcal{S}^c$ defined by $\mathcal{S}^c=\mathbb{I}\setminus \mathcal{S}$, with 
$\mathcal{S}=\left\{t\in \mathbb{I}: \left|s_t(\varphi^\star)\right|<1\right\}$. 
\end{itemize}
The expression of  $(\theta^\star,\varphi^\star)$ is then given by:
\begin{align}
	&\theta^\star=(XX^\top)^{-1}X\big(\bm{y}-\varphi^\star\big), \label{eq:unique1}
\end{align}
\begin{equation}\label{eq:unique2}
\begin{aligned}
			& \mbox{If } \left|\mathcal{S}^c\right|=0, \mbox{ then }\varphi^\star=0, \quad   \mbox{otherwise } \\
			&\varphi_{\mathcal{S}^c}^\star=\left(\Psi_{\mathcal{S}^c}^\top \Psi_{\mathcal{S}^c}\right)^{-1}\Psi_{\mathcal{S}^c}^\top\big(\Psi\bm{y}-\lambda s(\varphi^\star)\big), \: \: \varphi_{\mathcal{S}}^\star=0. 
\end{aligned}
\end{equation}
\end{lem}
\begin{pf}
$l(\theta,\varphi)$ is a quadratic function of $\theta$. For a fixed $\varphi^\star$, the minimizer  $\theta^\star$ of $l(\theta,\varphi^\star)$  is unique if and only if $X$ has full row rank, i.e., $\rank(X)=n$. The unique value of $\theta^\star$ is expressed in function of $\varphi^\star$ by \eqref{eq:unique1}. Plugging the expression  \eqref{eq:unique1} of $\theta^\star$ in the objective gives 
$$\tilde{l}(\varphi)\triangleq l(\theta^\star,\varphi)=\dfrac{1}{2}\left\|\Psi \bm{y}-\Psi\varphi\right\|_2^2+\lambda\left\|\varphi\right\|_1.$$
The rest of the proof then boils down to showing that the minimizer $\varphi^\star$ of $\tilde{l}(\varphi)$ is unique if and only if $\rank(\Psi_{\mathcal{S}^c})=\left|\mathcal{S}^c\right|$. To begin with, let us point out the following (see also\footnote{It is to be noted that the analysis in \cite{Tibshirani12} provides only a sufficient condition.} \cite{Tibshirani12}). If $\varphi^\star$ and $\xi^\star$ are two minimizers of $\tilde{l}(\varphi)$, then we have necessarily
\begin{align}
	&\Psi \varphi^\star=\Psi\xi^\star \label{eq:Psiphi}\\
	&s(\varphi^\star)=s(\xi^\star). \label{eq:s(phi)}
\end{align} 
The relation \eqref{eq:Psiphi} follows from the strict convexity of $\tilde{l}(\varphi)$ as a function of $\Psi \varphi$. In effect, by changing the optimization variable into $\delta=\Psi \varphi$, $\tilde{l}(\varphi)$ becomes 
$\frac{1}{2}\left\|\Psi \bm{y}-\delta\right\|_2^2+\lambda\left\|\Psi^\dagger\delta+v\right\|_1$, with $v$ a vector in $\ker(\Psi)$ and $\dagger$ referring to generalized inverse. This last function is strictly convex with respect to $\delta$. As a consequence, its minimizer is unique and equal to $\delta^\star=\Psi \varphi^\star$. 
To see why the relation \eqref{eq:s(phi)}  holds, plug the expression \eqref{eq:unique1} of $\theta^\star$ into \eqref{eq:opt-cond2-regularized-L1}. We get $\lambda s(\varphi^\star)=\Psi \bm{y}-\Psi\varphi^\star$. Combining this with \eqref{eq:Psiphi} (i.e., the uniqueness of $\Psi\varphi^\star$) yields immediately \eqref{eq:s(phi)}.   

Let us examine first the case where $\left|\mathcal{S}^c\right|=0$. This is indeed equivalent to $\mathcal{S}=\mathbb{I}$ and so, $\varphi^\star=0$. Would there exist another minimizer $\xi^\star$ of  $\tilde{l}(\varphi)$, it should obey \eqref{eq:s(phi)}, which implies that $\xi^\star$ is necessarily equal to zero. 

Now consider the case $\left|\mathcal{S}^c\right|>0$. \\
\textit{Sufficiency.}
Assume that $\rank(\Psi_{\mathcal{S}^c})=\left|\mathcal{S}^c\right|$. As argued above, any two minimizers $\varphi^\star$ and $\xi^\star$ of $\tilde{l}(\varphi)$ obey \eqref{eq:Psiphi}-\eqref{eq:s(phi)}. From \eqref{eq:s(phi)} we get that $\mathcal{S}\subset \left\{t\in \mathbb{I}:\xi_t^\star=0\right\}$, which implies that $\mathcal{S}^c\supset \supp(\xi^\star)$. As a consequence, we can write \eqref{eq:Psiphi} in the following reduced form $\Psi_{\mathcal{S}^c}(\varphi_{\mathcal{S}^c}^\star-\xi_{\mathcal{S}^c}^\star)=0. $
With $\rank(\Psi_{\mathcal{S}^c})=\left|\mathcal{S}^c\right|$, this implies that $\varphi^\star=\xi^\star$ and that the minimizer of 
$\tilde{l}(\varphi)$ is unique.

\textit{Necessity.} Assume that $\rank(\Psi_{\mathcal{S}^c})<\left|\mathcal{S}^c\right|$. Consider a nonzero vector $\eta\in \Re^N$ such that $\eta_{\mathcal{S}}=0$ and $\eta_{\mathcal{S}^c}\in \ker(\Psi_{\mathcal{S}^c})$. Let $\eta_1=\nu \eta$, with $\nu\neq 0$. It is straightforward to verify that $\Psi \varphi^\star=\Psi (\varphi^\star+\eta_1)$. Note that $\nu$ can be chosen sufficiently small such that $\varphi_t^\star$ and  $\varphi_t^\star+\eta_{1,t}$ have the same sign whenever $\varphi_t^\star\neq 0$. Following a similar path as in the proof of Theorem \ref{thm:equivalence}, we can establish that $s(\varphi^\star)=s(\varphi^\star+\eta_1)$. Finally, with $\Psi \varphi^\star=\Psi (\varphi^\star+\eta_1)$, $s(\varphi^\star)=s(\varphi^\star+\eta_1)$ and the fact that $\varphi^\star$ is an optimal solution (hence satisfying \eqref{eq:opt-cond2-regularized-L1}), it is easy to check that $\varphi^\star+\eta_1$ also satisfies \eqref{eq:opt-cond2-regularized-L1}. By Lemma \ref{lem1:opt-conditions-regularized}, $\varphi^\star+\eta_1$ ($\neq \varphi^\star$) solves \eqref{eq:regularized-L1}. 
Hence, the solution is not unique. 

\textit{Derivation of Eqs \eqref{eq:unique1}-\eqref{eq:unique2}.} These relations result from simple rearrangements of \eqref{eq:opt-cond1-regularized-L1}-\eqref{eq:opt-cond2-regularized-L1}. \qed
\end{pf}
\vspace*{-\baselineskip}
From Lemma \ref{lem2:opt-conditions-regularized}, it appears that the true vector $\bm{f}$ can be found by problem \eqref{eq:regularized-L1} if and only if there is a vector $\hat{\theta}\in \Re^n$ such that $(\hat{\theta},\bm{f})$ satisfies the conditions (i)-(ii) of Lemma \ref{lem2:opt-conditions-regularized}. In particular, $(\hat{\theta},\bm{f})$ must satisfy \eqref{eq:opt-cond2-regularized-L1}. A necessary condition for this is that $\Psi \bm{e} =\lambda s(\bm{f})$. And this implies that the regularization parameter must verify  $\lambda\geq \left\|\Psi \bm{e} \right\|_\infty$ when $\bm{f}=0$, and $\lambda= \left\|\Psi \bm{e} \right\|_\infty$    when $\bm{f}\neq 0$. Note further that if $\bm{e}= 0$  and $\bm{f}\neq 0$, then $\lambda$ must be equal to zero! 
However, if $\lambda$ is set to zero in \eqref{eq:regularized-L1}, then the solution set is 
$$\left\{(\theta,\varphi):\theta=(XX^\top)^{-1}X(\bm{y}-\varphi), \varphi \in \bm{y}+\im(X^\top)\right\}. $$
Since this set contains infinitely many elements, we conclude that it is unlikely to get exactly the true $\bm{f}$ by solving \eqref{eq:regularized-L1} irrespective of the value of the regularization parameter $\lambda$. \\
In any case, the estimation error can be bounded as follows. 
\begin{prop}
\label{prop:error-bound}
Assume that the conditions of Lemma \ref{lem2:opt-conditions-regularized} are satisfied and denote with $(\theta^\star,\varphi^\star)$ the  solution to problem \eqref{eq:regularized-L1}. Then 
{\setstretch{1} 
\begin{equation}\label{eq:error-bound}
		\left\|\theta^\star-\theta^o\right\|_2\leq (K_1\varepsilon+\lambda K_2) +K_1M \sqrt{\dfrac{\left|\mathcal{J}\cap\mathbb{I}^c(\theta^o)\right|}{\left|\mathcal{J}\right|}}
\end{equation}		
		where $\varepsilon = \max_{t\in \mathbb{I}}\left|e_t\right|$,   $M = \max_{t\in \mathbb{I}}\left|f_t\right|$, 		
\begin{equation}
K_1 = \max_{\left|J\right|\geq \nu_n(X)}\!\!\sqrt{\left|J\right|}\left\|(XX^\top)^{-1}(I_n+E_J+2E_J^2+E_J^3)X_J\right\|_2\label{eq:K1}			
\end{equation}
\begin{equation}
			K_2 = \max_{\left|J\right|\geq \nu_n(X)}\sqrt{\left|J\right|}\left\|(X_JX_J^\top)^{-1}X_J\right\|_2 \label{eq:K2}
\end{equation}
}
with 
$E_J=\left(X_{J^c}X_{J^c}^\top\right)\left(X_J X_J^\top\right)^{-1}$, $J\subset \mathbb{I}$. In \eqref{eq:error-bound} $I_n$ is the identity matrix of order $n$ ; the set $\mathcal{J}$ denotes the maximizing argument of \eqref{eq:K1} and $\mathbb{I}^c(\theta^o)=\mathbb{I}\setminus \mathbb{I}^0(\theta^o) $.  
\end{prop}
\begin{pf}
The idea of the proof consists in deriving first an expression of $\theta^\star-\theta^o$ and then working out a bound on its norm. 
From \eqref{eq:unique1} and the data model \eqref{eq:model}, we have
\begin{align*}
	\theta^\star 
	&=(XX^\top)^{-1}X\left(X^\top \theta^o+\bm{e}+\bm{f}-\varphi^\star\right)
\end{align*}
This, by noting that $\varphi_{\mathcal{S}}^\star=0$, can be written as 
$$(XX^\top)\left(\theta^\star-\theta^o\right)=X_{\mathcal{S}}\left(\bm{e}_{\mathcal{S}}+\bm{f}_{\mathcal{S}}\right) +X_{\mathcal{S}^c}\left(\bm{e}_{\mathcal{S}^c}+\bm{f}_{\mathcal{S}^c}-\varphi_{\mathcal{S}^c}^\star\right) $$
Using formula \eqref{eq:unique2} and manipulating a little, we arrive at
$$
\begin{aligned}
(XX^\top)&\left(\theta^\star-\theta^o\right)=\\
&\left[X_{\mathcal{S}}-X_{\mathcal{S}^c}(\Psi_{\mathcal{S}^c}^\top \Psi_{\mathcal{S}^c})^{-1}\Psi_{\mathcal{S}^c}^\top \Psi_{\mathcal{S}}\right]\left(\bm{e}_{\mathcal{S}}+\bm{f}_{\mathcal{S}}\right) \\
	&\quad +\lambda X_{\mathcal{S}^c}(\Psi_{\mathcal{S}^c}^\top \Psi_{\mathcal{S}^c})^{-1}\Psi_{\mathcal{S}^c}^\top s(\varphi^\star).
\end{aligned}
$$
Further calculations using  the Woodbury's matrix identity and exploiting the relation $Xs(\varphi^\star)=0$,  yield
$$\begin{aligned}
	(XX^\top)\left(\theta^\star-\theta^o\right)=
	&\left(I_n+E_\mathcal{S}+2E_\mathcal{S}^2+E_\mathcal{S}^3\right)X_{\mathcal{S}} \left(\bm{e}_{\mathcal{S}}+\bm{f}_{\mathcal{S}}\right)\\
	&-\lambda\left(I+E_\mathcal{S}\right)X_{\mathcal{S}} s(\varphi_{\mathcal{S}}^\star) 
\end{aligned}$$
with $E_\mathcal{S}=\left(X_{\mathcal{S}^c}X_{\mathcal{S}^c}^\top\right)\left(X_{\mathcal{S}}X_{\mathcal{S}}^\top\right)^{-1}$.
The result follows by multiplying with $(XX^\top)^{-1}$, remarking that $(XX^\top)^{-1}=\left(X_{\mathcal{S}}X_{\mathcal{S}}^\top\right)^{-1}(I_n+E_\mathcal{S})^{-1}$ and taking the euclidean norm. 
\end{pf}
It is interesting to notice that the numbers $K_1$, $K_2$ and $\left|\mathcal{J}\right|$ depend solely on the data matrix $X$. Moreover, when the sequence $\left\{f_t\right\}$ contains only a few nonzero elements (but otherwise arbitrarily large), the last term in \eqref{eq:error-bound} is likely to vanish. As a consequence, even though the bound $M$ can be large in principle, the bound on the estimation error can be kept at a reasonable level. 

\section{Extension to multivariable systems}\label{sec:multivariable}
We consider now the multivariable analogue of model \eqref{eq:model} written in the  form
\begin{equation}\label{eq:multivariable}
	y_t=A^ox_t+f_t+e_t,
\end{equation}
where $y_t\in \Re^m$ is the output vector at time $t$, $\left\{f_t\right\}\subset \Re^m$ is the sequence of errors, $\left\{e_t\right\}\subset \Re^m$ is the noise sequence, $A^o\in \Re^{m\times n}$ is the parameter matrix. 

The question of interest is how to recover the matrix $A^o$ from measurements corrupted by a vector sequence of sparse errors $\left\{f_t\right\}$.  The sparse optimization approach is still applicable to this case, that is, we can formulate the estimation problem as 
\begin{equation}\label{eq:Ao-L0}
	\minimize_{A\in \Re^{m\times n}}\Big|\Big\{t: y_t-A x_t\neq 0\Big\}\Big|
\end{equation}
with $|\cdot|$ standing for cardinality. It can be easily verified that Theorem \ref{thm:Uniqueness_L0} applies to  \eqref{eq:Ao-L0} as well.  

The convex relaxation takes the form of a nonsmooth optimization with a cost functional consisting of a sum-of-norms of errors \cite{Ohlsson2010a,Chen11-IFAC}, 
\begin{equation}\label{eq:matrix-relax}
	\minimize_{A\in \Re^{m\times n}} \sum_{t=1}^N\left\|y_t-Ax_t\right\|_2 
\end{equation}
with $\left\|\cdot\right\|_2$ referring to the Euclidean norm. 
\begin{thm}\label{thm:equivalence2}
A matrix $A^\star\in \Re^{m\times n}$ solves the sum-of-norms problem \eqref{eq:matrix-relax},  if and only if any of the following equivalent statements holds: 
\begin{enumerate} 
\item[\emph{T1.}] There exists a sequence of vectors $\left\{\beta_t\right\}_{t\in \mathbb{I}^0(A^\star )}\subset \mathcal{B}_2(0,1)$ such that 
\begin{equation}\label{eq:iff-cond-T1}
	\sum_{t\notin \mathbb{I}^0(A^\star)} v_t^\star x_t^\top +\sum_{t\in \mathbb{I}^0(A^\star )}\beta_t x_t^\top=0,
\end{equation}
where $v_t^\star = (y_t-A^\star x_t)/\left\|y_t-A^\star x_t\right\|_2$. Here, $\mathcal{B}_2(0,1)\subset \Re^m$ is  the Euclidean unit ball of $\Re^m$.
\item[\emph{T2.}] For any matrix $\Lambda \in \Re^{m\times n}$, 
\begin{equation}\label{eq:iff-cond-T2}
	\Big|\sum_{t\notin \mathbb{I}^0(A^\star)}{v_t^\star}^\top \Lambda x_t\Big|\leq \sum_{t\in \mathbb{I}^0(A^\star)}\big\|\Lambda x_t\big\|_2.
\end{equation}
 \item[\emph{T3.}] The optimal value of the problem 
\begin{equation}\label{eq:min_alpha-2}
\begin{aligned}
		&\min_{Z\in \Re^{m\times p}}\:\left\|Z\right\|_{2,\infty} 
		\: \mbox{\emph{subject to} }\:  V^\star X_{\mathbb{I}^c(A^\star)}^\top=ZX_{\mathbb{I}^0(A^\star)}^\top
\end{aligned}
\end{equation}
 $p=\left|\mathbb{I}^0(A^\star)\right|$ and $V^\star$ being a matrix formed with the unit 2-norm vectors $v_t^\star$, for $t\in \mathbb{I}\setminus\mathbb{I}^0(A^\star)$, \\
\emph{is  smaller than $1$}.
\end{enumerate}

\noindent Moreover,  the solution $A^\star$ is unique if and only if any of the following assertions is true:
\begin{enumerate}
	\item[\emph{T1'.}] \eqref{eq:iff-cond-T1} holds and $\rank(X_\mathcal{T})=n$ where  $\mathcal{T}=\left\{t\in \mathbb{I}^0(A^\star): \left\|\beta_t\right\|_2<1\right\}$. 
	\item[\emph{T2'.}] \eqref{eq:iff-cond-T2} holds with strict inequality symbol for all $\Lambda\in \Re^{m\times n}$, $\Lambda\neq 0$.
\end{enumerate}
\end{thm}
\begin{pf}
The proof is similar to that of Theorem \ref{thm:equivalence}. It is therefore omitted here. 
\end{pf}
It is interesting to note that based on Theorem \ref{thm:equivalence2}, the analysis carried out in the previous sections can be easily generalized to the multivariable case.  In particular, Proposition \ref{prop:infinite-errors} and  Theorems \ref{thm:X'(XX')}-\ref{prop:sufficient-condition} can be restated for the multivariable model \eqref{eq:multivariable} with only some slight modifications. For illustration purpose, we just state below the multivariable counterpart of  Corollary \ref{cor:NS-Conds}. 
\begin{cor}\label{thm:sufficient-multivariable}
Let $d$ be an integer.
Then the following three statements are equivalent. 
\begin{enumerate}
\item[\emph{(j)}]
\begin{multline}
	\!\!\! \!\!\!\! \forall \: A\in \Re^{m\times n}, \forall Y\in \Re^{m\times N}, \: \left|\mathbb{I}^0\left(A\right)\right|\geq d\\ \quad \Rightarrow A\in \argmin_{W\in \Re^{m\times n}}
	\left\|Y-WX\right\|_{2,\col}
\end{multline}
	\item[\emph{(jj)}] 
	\begin{equation}\label{eq:(jj)}
	\begin{aligned}
\!\!\! \!\!\!\!\!\!\!\!\!	\max_{\substack{(I,I^c):\\\left|I\right|=d}}\: \: \max_{\Lambda\in \Re^{m\times n}}&\left\{\left\|\Lambda X_{I^c}\right\|_{2,\col}\suchthat \: \left\|\Lambda X\right\|_{2,\col}=1\right\} \leq 1/2
	\end{aligned}
\end{equation}  
\item[\emph{(jjj)}] 
\begin{equation}\label{eq:(jjj)}
	\begin{aligned}
	\!\!\!\!\!\!\!\!\!\!\!	\max_{\substack{(I,I^c):\\\left|I\right|=d}} \max_{V\in \mathbb{B}^{m\times |I^c|}} \min_{Z\in \Re^{m\times |I|}} &\left\{\left\|Z\right\|_{2,\infty} \suchthat\right.  \\
		&  \Big. \: X_{I^c}V^\top =X_IZ^\top \Big\}\leq 1
	\end{aligned}
\end{equation}
with $\mathbb{B}^{m\times q}=\Big\{\big[\begin{matrix}b_1 & \cdots & b_q\end{matrix}\big]\in \Re^{m\times q}, \: b_i\in \mathcal{B}_2(0,1)\Big\}$.
\end{enumerate} 
\end{cor}
\begin{pf}
The proof is similar to that of Corollary \ref{cor:NS-Conds}.  \\
(jj) $\Leftrightarrow$ (jjj) : We exploit the equivalence between \eqref{eq:iff-cond-T2} and \eqref{eq:min_alpha-2}. 
First by letting $I=\mathbb{I}^0(A^\star)$, $I^c=\mathbb{I}^c(A^\star)$, $V_{I^c}$ be a matrix collecting all the vectors $v_t^\star\in \mathbb{B}^m$, $t\in I^c$,  \eqref{eq:iff-cond-T2} can equivalently be written as
$$ 
\begin{aligned}
	\max_{\Lambda\in \Re^{m\times n}}&\Big[\big|\tr\big(V_{I^c}^\top \Lambda X_{I^c}\big)\big|+\left\|\Lambda X_{I^c}\right\|_{2,\col}\suchthat \left\|\Lambda X\right\|_{2,\col}=1\Big]\\
	&\leq 1
\end{aligned}
$$
Maximizing over all the sets $(I,I^c)$ satisfying $\left|I\right|=d$ and over all $V \in \mathbb{B}^{m\times |I^c|}$ yields \eqref{eq:(jj)} after remarking that $\max_{V_{I^c}\in \mathbb{B}^{m\times |I^c|}}\big|\tr\big(V_{I^c}^\top \Lambda X_{I^c}\big)\big|=\left\|\Lambda X_{I^c}\right\|_{2,\col}$. 
Proceeding similarly from \eqref{eq:min_alpha-2}, yields  \eqref{eq:(jjj)}. 
Hence (jj) $\Leftrightarrow$ (jjj).\\
(j) $\Leftrightarrow$ (jj) : By Theorem \ref{thm:equivalence2}  and the first part of the proof, any matrix $A$ with $\left|\mathbb{I}^0(A)\right|=d$ minimizes the objective $\left\|Y-WX\right\|_{2,\col}$ (with variable $W$) if and only if \eqref{eq:(jj)} holds. The conclusion is obtained by observing that 
$$\max_{\substack{(I,I^c):\\\left|I\right|=d}}\max_{\Lambda\in \Re^{m\times n}}\left\{\left\|\Lambda X_{I^c}\right\|_{2,\col}\suchthat \: \left\|\Lambda X\right\|_{2,\col}=1\right\}$$ 
is decreasing as a function of $d$.\qed 
\end{pf}

An analogue of Corollary \ref{cor:NS-Conds-2} can be obtained similarly. 
It is interesting to note that the statement \eqref{eq:sufficient} of Theorem \ref{prop:sufficient-condition} holds unchanged in the multivariable case with $p=\xi(X)$. 
\begin{rem}[Geometric median]
In the special case where $n=1$, the matrix $A^o$ in \eqref{eq:multivariable} reduces to a vector $a^o\in \Re^m$. Assuming $x_t=1$ for all $t$, the problem \eqref{eq:matrix-relax} then becomes 
\begin{equation}\label{eq:geo-median}
	\minimize_{a\in \Re^m} \sum_{t=1}^N\left\|y(t)-a\right\|_2.
\end{equation}
This is the so-called geometric median problem. \\
By applying \eqref{eq:iff-cond-T2}, we can see that $a^o$ solves \eqref{eq:geo-median} if
$|\mathbb{I}^0(a^o)|/N\geq 1/2$. 
\end{rem}

\section{Numerical illustration}\label{sec:simulation}

\subsection{Static models subject to intermittent gross errors}\label{subsec:static}
In our first experiment we consider static linear and affine models of
the form \eqref{eq:model} with $n=4$ and $N=500$. The affine model
refers to the case where the regressor $x_t$ has the form
$x_t=[\tilde{x}_t^\top\; 1]^\top$.  The goal is to estimate the
probability of exact recovery of the true parameter vector by problem
\eqref{eq:L1-Cost} in function of the number of nonzero elements in
the sequence $\left\{f_t\right\}$. For this purpose, the noise
$\left\{e_t\right\}$ is set to zero. The nonzero elements of
$\left\{f_t\right\}$ are drawn from a Gaussian distribution with mean
$100$ and variance $1000^2$.  For each level of sparsity (\ie proportion
of nonzeros), a Monte Carlo simulation of size $100$ is carried out
with randomly generated static/affine models and $500$ data samples at
each run. Repeating this for four situations (linear/affine and
linear/affine with positive $f_t$'s), we obtain the results depicted
in Figure \ref{fig:proba}. We observe that in the linear case, problem
\eqref{eq:L1-Cost} provides the true parameter vector when the output
is affected by up to $80\%$ of nonzero gross errors. This is because
the data $\left\{x_t\right\}$ which were sampled from a Gaussian
distribution are very generic. 
In the case of affine models, the performance is a little less good. If we set all $f_t$'s to have the same sign, then as suggested by condition \eqref{eq:necessary-affine},  the percentage of outliers that can be corrected  by the optimization problem \eqref{eq:L1-Cost} cannot exceed $50\%$. 
\begin{figure*}
\centering
\subfloat[Linear static model]{\includegraphics[width=7.5cm,height=5.2cm]{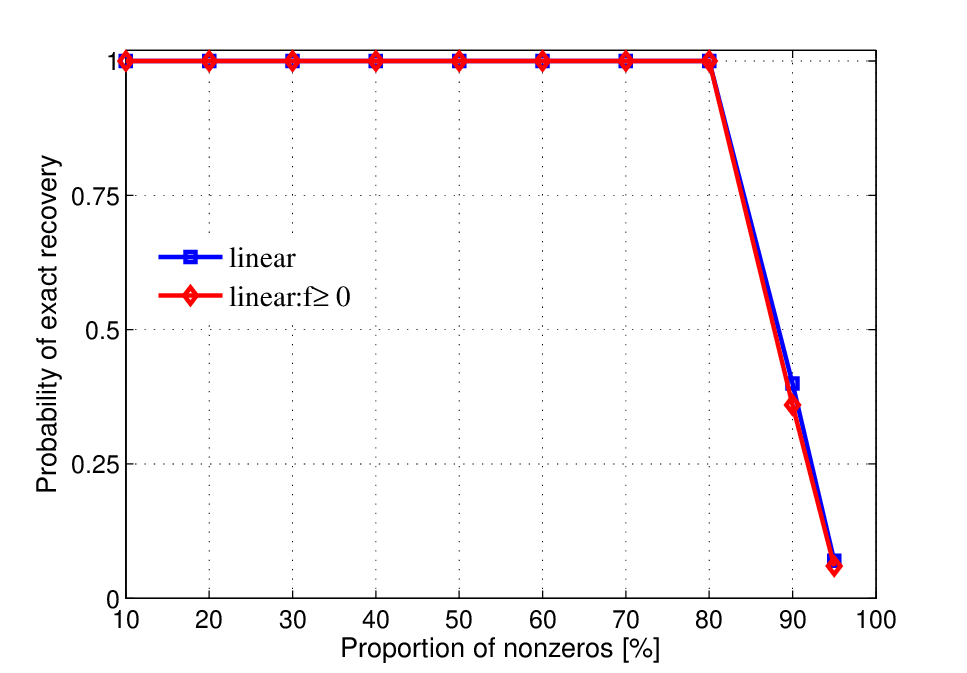}}
\subfloat[Affine static model]{\includegraphics[width=7.5cm,height=5.2cm]{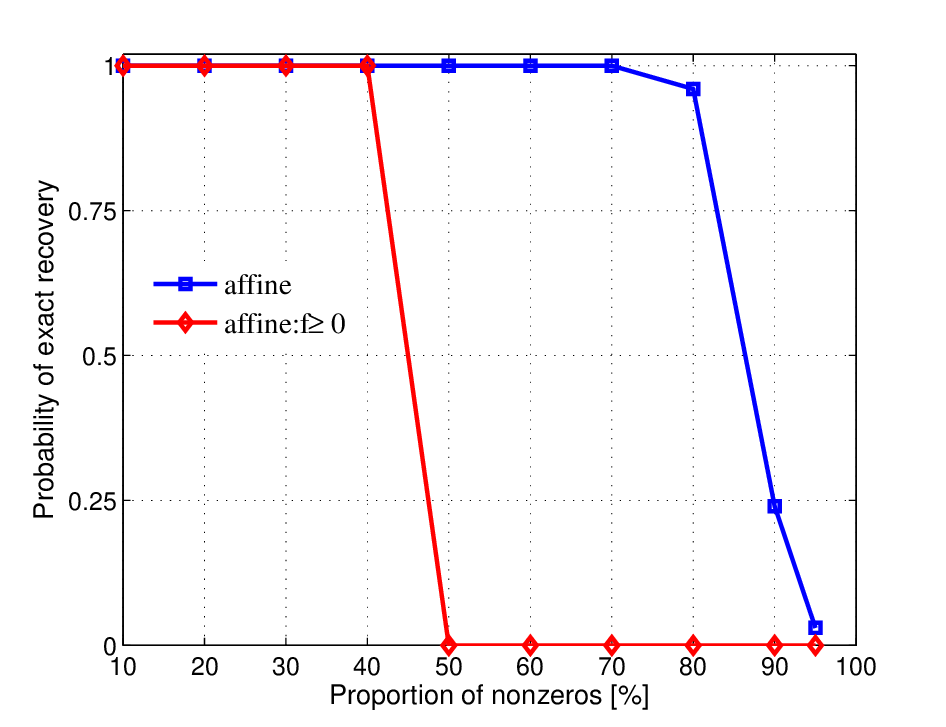}}
\caption{Estimates of probabilities of exact recovery when noise $\left\{e_t\right\}$ is equal to zero. From a numerical point of view, the recovery is said to be exact if $\big\|\hat{\theta}-\theta^o\big\|_2\leq 10^{-5}$ and inexact otherwise, with $\hat{\theta}$ being the estimated parameter vector.}
\label{fig:proba}
\end{figure*}
\subsection{Static models with both noise and gross errors}\label{subsec:static-2}
Consider now the case of static models of the form \eqref{eq:model} in
the presence of both $\left\{e_t\right\}$ and $\left\{f_t\right\}$
sampled from Gaussian distributions $\mathcal{N}(0,\sigma_e^2)$ and
$\mathcal{N}(100,1000^2)$ respectively. The variance $\sigma_e^2$ is
selected so as to achieve a certain signal to noise ratio 
before the gross error sequence is added to the output. Again, by
carrying out a Monte-Carlo simulation of size $100$ with different
sparsity levels and randomly generated models at each run, we obtain
the average errors plotted in Figure \ref{fig:error}. It turns out
that the results returned by problems \eqref{eq:L1-Cost} and
\eqref{eq:regularized-L1} with $\lambda=0.10$ are almost the same for an SNR in $\left\{10\mbox{ dB}, 20\mbox{ dB}\right\}$.
The performance can be assessed by comparing with an "oracle" estimate \ie the least squares estimate one would obtain if the locations of zeros in the sequence $\left\{f_t\right\}$ were known. The results in Figure \ref{fig:error} tend to suggest that the proposed approach performs very well. For the current numerical experiment, our results are very close to the ideal estimate when the proportion of nonzeros is less than $70\%$. 
\begin{figure*}
\centering
\subfloat[Static model: SNR = 20 dB]{\includegraphics[width=8cm,height=5cm]{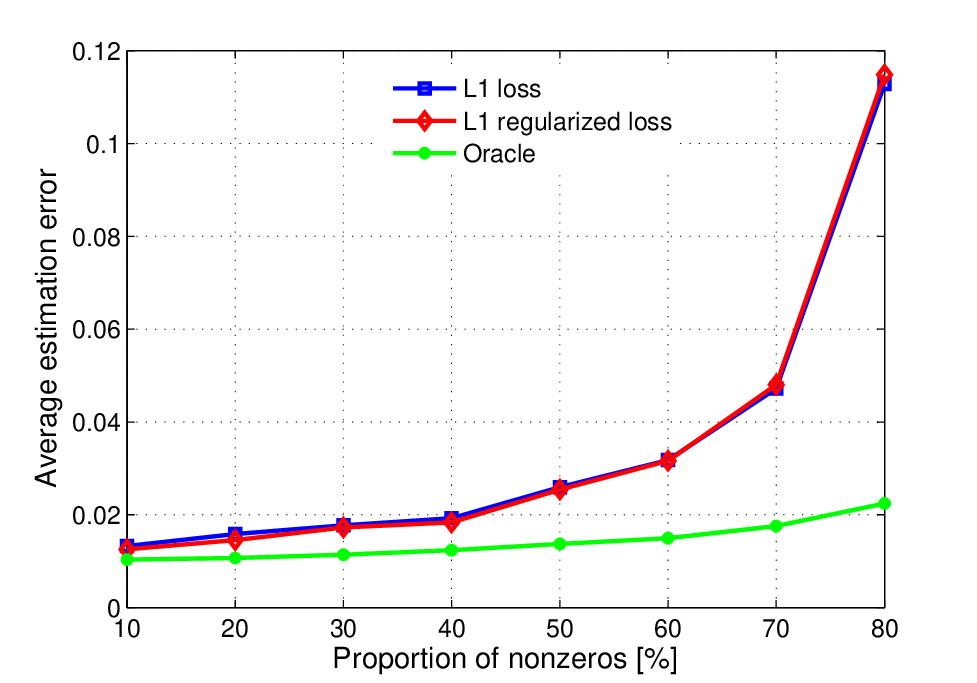}}
\subfloat[Static model: SNR = 10 dB]{\includegraphics[width=8cm,height=5cm]{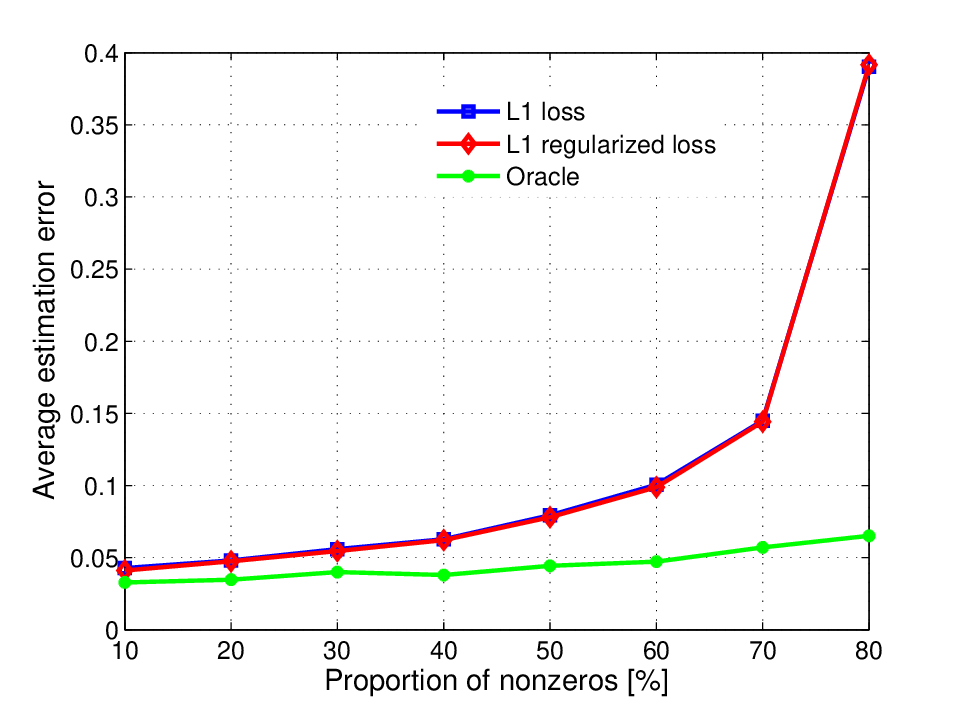}}
\caption{Average relative estimation error versus sparsity level.}
\label{fig:error}
\end{figure*}
\subsection{Dynamic linear models subject to sensor intermittent faults}
In the case when \eqref{eq:model} represents a dynamic ARX model
subject to gross errors, it can be observed (see
Fig. \ref{fig:proba2}) that the probabilities of exact recovery are
much smaller than in the static case studied in Section
\ref{subsec:static}. This difference is related to the richness (or
genericity) 
of the regression vectors (columns of $X$) involved in each case. In the static example above, the vectors $\left\{x_t\right\}$ are freely sampled in any direction of $\Re^n$ by following a Gaussian distribution. In the dynamic system case however, the data vectors $\left\{x_t\right\}$ constructed as in \eqref{eq:Regressor} are constrained to lie on a manifold. As a result, the 
data matrix $X$ generated by the dynamic system is less generic. According to conditions of the paper, and \eqref{eq:sufficient} in particular, there is a threshold depending on the richness of the data such that exact recovery is guaranteed whenever the number of zero entries in $\bm{f}$ is larger than this threshold. So, the more generic the data contained in $X$ are, the more outliers can be removed by problem \eqref{eq:L1-Cost}.  
Note that the lack of sufficient genericity can be compensated (to some extent) by implementing the iterative sparsity enhancing technique  ($\ell_1$ reweighted algorithm) described in Section \ref{subsec:Sparsity-Enhancing}. This leads, for only two iterations, to significantly improved results as represented in Figure \ref{fig:L1reweigthed}.  
\begin{figure*}
\centering
\subfloat[Linear dynamic model]{\includegraphics[width=8cm,height=5cm]{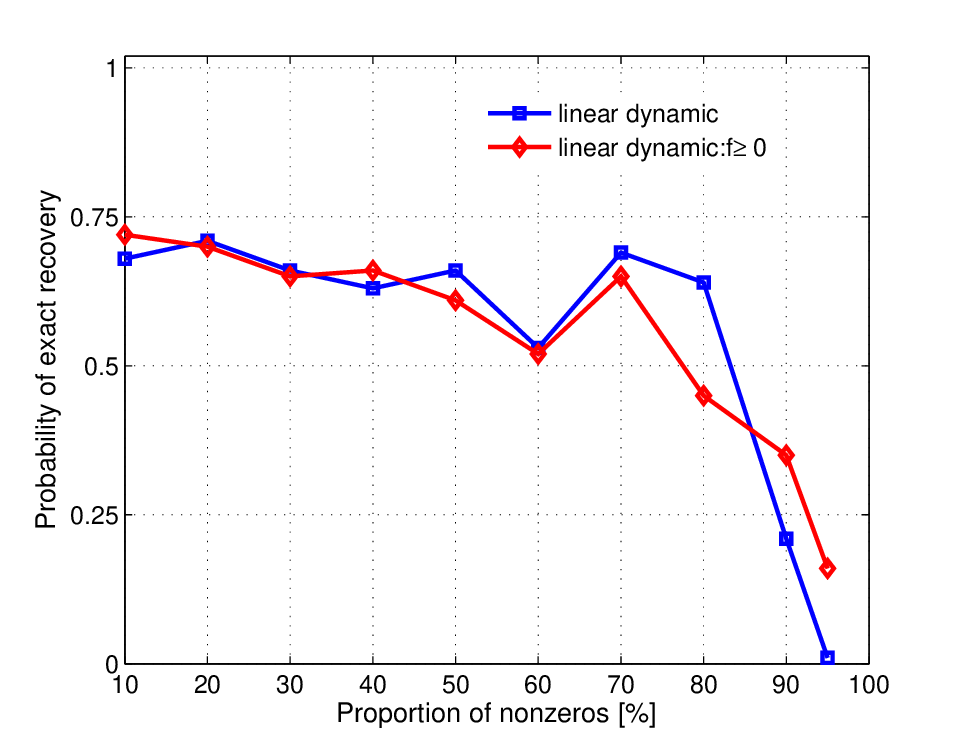}}
\subfloat[Affine dynamic model]{\includegraphics[width=8cm,height=5cm]{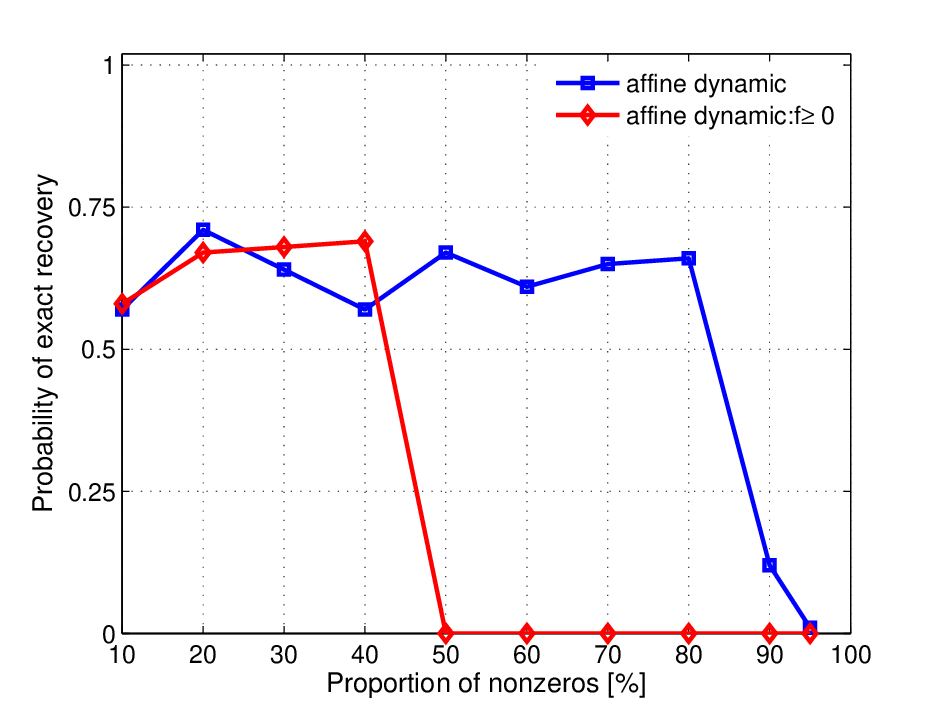}}
\caption{Estimates of probabilities of exact recovery when noise $\left\{e_t\right\}$ is equal to zero. Results of a Monte-Carlo simulation of size $100$ with randomly generated linear ARX systems of order $n_a=n_b=2$.}
\label{fig:proba2}
\end{figure*}
\begin{figure}
\centering
\includegraphics[width=8.5cm,height=5cm]{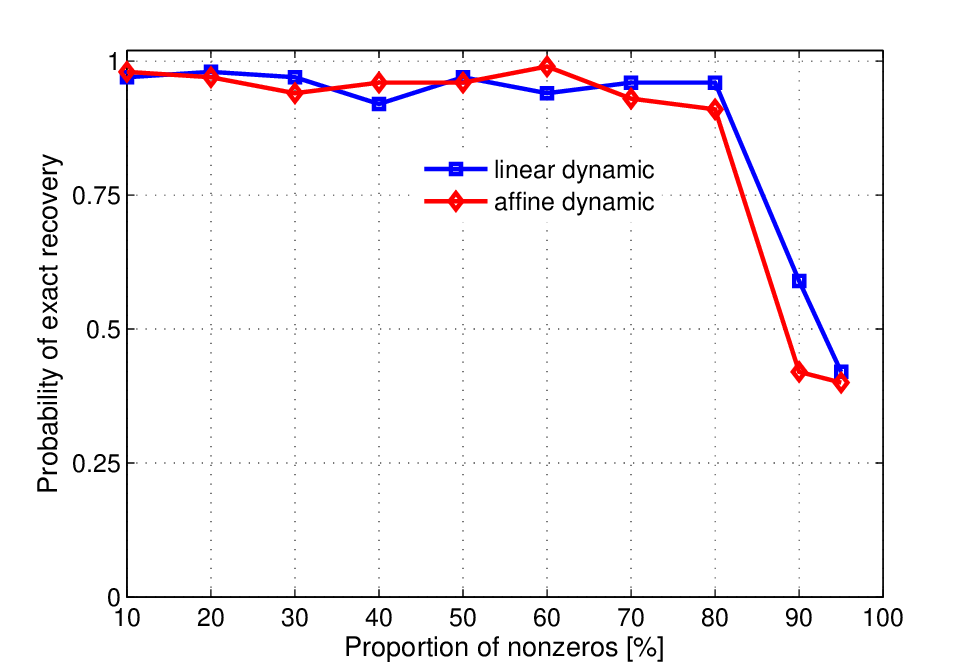}
\caption{Estimates of probabilities of exact recovery by reweigthed $\ell_1$ minimization when noise $\left\{e_t\right\}$ is equal to zero. Results of a Monte-Carlo simulation of size $100$ with randomly generated linear ARX systems with orders $n_a=n_b=2$.}
\label{fig:L1reweigthed}
\end{figure}

\subsection{Numerical evaluation of sufficient bounds}\label{subsec:bounds}
This subsection presents a numerical evaluation of the estimates of number of outliers that can be corrected by the nonsmooth optimization-based estimator. Note that the numbers $k_i(X)$, $i=1,2$ from Theorem \ref{thm:X'(XX')} are hard to compute numerically because this would require a combinatorial optimization. To be more specific, the complexity of evaluating literally $k_i(X)$, $i=1,2$, is about $\sum_{k=1}^{k_i(X)}\binom{N}{k}C_i(N,n,k)$ where $\binom{N}{k}$  refers to the binomial coefficient,  $C_1(N,n,k) =\text{O}\big(n^3+kn(N+2n-k)\big)$ and $C_2(N,n,k) =\text{O}\big((n^2+nN)(N+n-k)\big)$ denote the complexity induced by the computation of $\|X_I^\top(X_IX_I^\top)^{-1}X_{I^c}\|_\infty$ and $\|X_{I^c}^\top(XX^\top)^{-1}X\|_1$ respectively with  $I$ and $I^c$ being some sets such that $\left|I\right|=k\leq N$ and $|I^c|=N-k$. \\ Therefore we just compare those bounds which are easier to compute. More specifically, four thresholds are compared:
\begin{itemize}
	\item The bounds $1/(2r(X))$ and $1/2+1/(2\xi(X))$ obtained in Theorem \ref{prop:sufficient-condition}.
	\item The mutual coherence-based bound $1/2\left(1+1/\mu(P_X)\right)$ with $P_X=I_N-X^\top\left(XX^\top\right)^{-1}X$ obtained in \cite{Bruckstein09,Bako11-Automatica}. Here $\mu$ represents the so-called mutual coherence. 
	\item  The bound $T(X)$ \cite{Sharon09-ACC} which  is used as a reference since it corresponds indeed to a direct computation of $N-\pi^o(X)+1$ (assuming the inequality in \eqref{eq:(ii)} is replaced with a strict one), see Eq.  \eqref{eq:Pi-O}. Recall that computing such a bound has a combinatorial complexity in the dimensions $(n,N)$ of the matrix $X$. Therefore, to make it feasible at a reasonable time on a standard computer, we have to set $n=2$ and $N\leq 200$.   
\end{itemize}
Figure \ref{fig:bounds-comparison-static} compares the sufficient thresholds in the case of static data drawn from a Gaussian distribution $\mathcal{N}(0,I_3)$. Figure \ref{fig:bounds-comparison-dynamic} compares the same thresholds for dynamic data in the form \eqref{eq:Regressor}. The generating system in this case is an ARX model defined by $y_t=-0.40y_{t-1}-0.15 u_{t-1}$ and driven by a normally distributed input sequence. In all cases, the data matrix $X$ is normalized so as to have unit $\left\|\cdot\right\|_{\Sigma_X}$-norm columns before being processed. Here, the norm $\left\|\cdot\right\|_{\Sigma_X}$ is   defined by $\left\|x\right\|_{\Sigma_X}=(x^\top \Sigma_X^{-1}x)^{1/2}$ with $\Sigma_X=XX^\top$. 
 The plots in Figure \ref{fig:bounds-comparison-static} and Figure \ref{fig:bounds-comparison-dynamic} draw the average values obtained over $100$ independent runs in term of percentage with respect to the total number of data.  
The results suggest three interesting facts :
\begin{itemize}
\item All the bounds are very loose that is, they largely underestimate the number of admissible gross errors. For example Figure \ref{fig:proba} shows that exact recovery can be achieved in the face a relatively large proportion (more than $70\%$) of corrupted data while the sufficient bounds in Figure \ref{fig:bounds-comparison-static} indicate a value around $20$. This is normal since the bounds reflect worst-case distributions of the outliers and their signs (see Theorem \ref{thm:equivalence} and Corollary \ref{cor:NS-Conds}).    
\item The bound based on $\xi(X)$ approaches the bound $T(X)$ \cite{Sharon09-ACC} while still enjoying  less numerical complexity. 
The other bounds based respectively on mutual-coherence and $r(X)$ are overall very close. These two last bounds seem to be more sensitive to the richness of the data and probably to their magnitudes also.  This fact is more apparent when the data are not normalized. 
\item As could be intuitively expected, the dynamic data generated by a linear system are less generic. The bounds obtained in this case are smaller. The question as to which type of dynamic system can generate more generic data is open. 
\end{itemize}

\begin{figure}[h]
\psfrag{Tsh}{\scriptsize bound in \cite{Sharon09-ACC}}
\psfrag{Tinf}{\scriptsize$\xi$-based}
\psfrag{Tmu}{\scriptsize$\mu$-based}
\psfrag{Tr}{\scriptsize$r$-based}
\includegraphics[width=9cm,height=5.5cm]{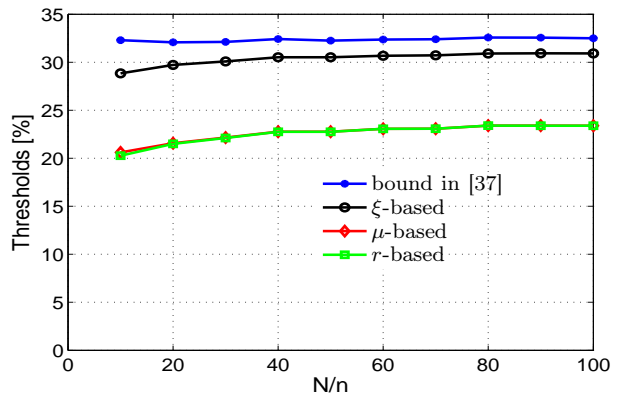}
\caption{Comparison of sufficient bounds on the number of gross errors for exact recovery : static data drawn from a Gaussian distribution.}
\label{fig:bounds-comparison-static}
\end{figure}
\begin{figure}[h]
\psfrag{Tsh}{\scriptsize bound in \cite{Sharon09-ACC}}
\psfrag{Tinf}{\scriptsize$\xi$-based}
\psfrag{Tmu}{\scriptsize$\mu$-based}
\psfrag{Tr}{\scriptsize$r$-based}
\includegraphics[width=9cm,height=5.5cm]{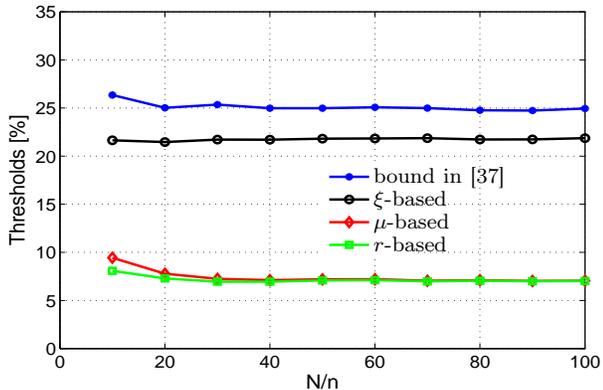}
\caption{Comparison of sufficient bounds on the number of gross errors for exact recovery : dynamic data generated by an ARX model.}
\label{fig:bounds-comparison-dynamic}
\end{figure}

\paragraph{Comparison of execution times}
Evaluating the $r$-based and  $\mu$-based bounds defined above is clearly very cheap as compared to the two other bounds. 
Therefore we shall only compare the execution times for the bounds $T(X)$ \cite{Sharon09-ACC} and $1/2(1+1/\xi(X))$ (see Theorem \ref{prop:sufficient-condition}) for $X\in \Re^{n\times N}$. This is done by measuring the average time over $10$ runs.\footnote{The computation is performed in a Matlab environment (version 2013a,64-bit), on a computer equipped with a processor Intel(R) Core(TM) i7-3630QM CPU@2.4Ghz, RAM 16Go. Only $10$ runs have been considered because the computation time for $T(X)$ grows very quickly beyond the capacity of the computer. Indeed $T(X)$ is computed only once when $n=5$ because the algorithm takes too long to complete (about 2 hours for each run in this case). Note that what matters in this experiment is not the numerical values of the execution times but the trend they exhibit.} The results reported in Table \ref{tab:complexity} show that for a given number $N$ of data points, the computation time  for the algorithm in \cite{Sharon09-ACC} is small for $n\leq 3$ but grows very fast (at a combinatorial rate) when $n$ increases.  In contrast, the cost associated with the evaluation of the $\xi$-based bound grows only at a logarithmic rate. This numerical experiment confirms that the proposed $\xi$-based bound is algorithmically cheaper to compute than $T(X)$. For example, for  $X\in \Re^{5\times 200}$ (see last column of Table \ref{tab:complexity}), computing $T(X)$ takes nearly 2 hours while the $\xi$-based bound derived in the current paper is obtained in less than $27$ seconds.  
\begin{table}[h]
	\centering
		\begin{tabular}{|c|c|c|c|c|c|}
		\hline 
			$n$        & 2 & 3 &4 & 5 \\ \hline
			$T(X)$  [sec.]& 0.02 & 1.40 & 101.82 & $5.62\times 10^3$\\ \hline
			 $\xi$-based [sec.] &22.70 & 24.56 & 25.44 & 26.80 \\  \hline 
		\end{tabular}
		\caption{Empirical comparison of the numerical complexities associated with evaluating $T(X)$ \cite{Sharon09-ACC} and the $\xi$-based bound (see Theorem \ref{prop:sufficient-condition}) in term of execution times. Here the number $N$ of data points is fixed and equal to $200$. }
		\label{tab:complexity}
\end{table}
\section{Conclusion}\label{sec:conclusion}
In this paper we have discussed the potential of nonsmooth convex optimization for addressing the problem of robust estimation. Considering in particular the problem of inferring an unknown parameter vector from measurements which are subject to  possibly large gross errors, we have shown that an exact recovery is possible regardless of the number of gross errors provided certain conditions of genericity hold. Then we investigated worst-case conditions which depend solely on the number of gross errors affecting the data. Necessary and sufficient conditions have been derived in this case. Since such conditions are numerically expensive to test directly, we have relaxed them into some sufficient but relatively tight conditions for  exact recovery. Simulations results reveal that the proposed worst-case  conditions for exact recovery are somewhat pessimistic when compared to the potential of the nonsmooth estimator in practice.  Concerning the identification problem, future work will consider the problem of designing the excitation of a dynamic system so as to achieve such strong genericity properties on the regressor matrix.

\begin{ack}
The authors thank the anonymous reviewers and the associate editor for insightful comments.
\end{ack}
\bibliographystyle{abbrv}

\end{document}